\documentclass[11pt, copyright,creativecommons]{eptcs}
\providecommand{\event}{EXPRESS 2010} % Name of the event you are submitting to
\usepackage{breakurl}             % Not needed if you use pdflatex only.

\title{Process Behaviour: Formulae versus Tests\\ (Extended Abstract)}
\author{Andrea Cerone
\institute{Trinity College Dublin\\ Dublin, Ireland}
\institute{School of Computer Science and Statistics\thanks{The financial support of Science Foundation Ireland is gratefully appreciated.}}
\email{ceronea@cs.tcd.ie}
\and
Matthew Hennessy
\institute{Trinity College Dublin\\ Dublin, Ireland}
\institute{School of Computer Science and Statistics}
\email{Matthew.Hennessy@cs.tcd.ie}
}
\def\titlerunning{Process Behaviour: Formulae v. Tests}
\def\authorrunning{A, Cerone \& M. Hennessy}

\usepackage{processes}
\usepackage{xspace}
\usepackage{enumerate}
\usepackage{color}
\usepackage{url}
\usepackage{prooftree}

\usepackage{txfonts}
\definecolor{what}{rgb}{.40,.10,.10}

\definecolor{greenish}{rgb}{.15,.6,.25} % Only RGB, CMYK and BW are predefined.
\newcommand{\Rmay}{\;R_{\scriptstyle{\text{may}}}\;}

\begin{document}
%\makecover
\maketitle
\begin{abstract}
  Process behaviour is often defined either in terms of the tests they
  satisfy, or in terms of the logical properties they enjoy. Here we
  compare these two approaches, using \emph{extensional testing} in
  the style of DeNicola, Hennessy, and a recursive version of the
  property logic HML.

We first characterise subsets of this property logic which can be
captured by tests.  Then we show that those subsets of the property
logic capture precisely  the power of tests.
\end{abstract}
%\tableofcontents
\section{Introduction}

One central concern  of concurrency theory is to determine whether two processes 
exhibit the same behaviour; to this end, many notions of behavioural equivalence 
have been investigated \cite{rob}. One approach, proposed in \cite{dhn}, is 
based on tests. Intuitively two processes are \emph{testing equivalent}, 
$p \approxtest q$, relative to a set of tests $T$ if $p$ and $q$ pass 
exactly the same set of tests from $T$. Much here depends of course on
details, such as the nature of tests, how they are applied and how
they succeed. 

In the framework set up in \cite{dhn} observers 
have very limited ability to manipulate the processes
under test; informally processes are conceived as completely
independent entities who may or may not react to testing requests;
more importantly the application of a test to a process simply
consists of a run to completion of the process in a \emph{test
  harness}. Because processes are in general nondeterministic,
formally this leads to two testing based equivalences, $ p \approxmay
q$ and $p \approxmust q$; the latter is determined by the set of tests
a process guarantees to pass, written $p \mustsatisfy t$,
while the former by those it is possible to
pass, $p \maysatisfy t$.  The \emph{may} equivalence provides a basis for the so-called
trace theory of processes \cite{csp} , while the \emph{must} equivalence
can be used to justify the various denotational models based on  \emph{Failures}
used in the theory of CSP, \cite{csp,olderog,rocco}.

Another approach to behavioural equivalence is to say that two
processes are equivalent unless there is a property which one enjoys
and the other does not.  Here again much depends on the chosen set of
properties, and what it means for a process to enjoy a property.
\textit{Hennessy Milner Logic} \cite{hml} is a modal logic often used
for expressing process properties in term of the actions they are able
to perform.  It is well-known that it can be used, via differing
interpretations, to determine numerous variations on
\emph{bisimulation equivalence}, \cite{ccs,luca}. What has received
very little attention in the literature however is the relationship
between these properties and tests.  This is the subject of the current
paper.

More specifically, we address the question of determining which
formulae of a recursive version of the Hennessy Milner Logic, which we
will refer to as \rechml, can be used to characterise tests.  This
problem has already been solved in \cite{aceto} for a non-standard
notion of testing; this is discussed more fully later in the paper.
But we will focus on the more standard notions of \emph{may} and
\emph{must} testing mentioned above.

% Another important aspect of concurrency theory is that of 
% formalising properties which are of interests for processes. 
% This is usually done by introducing a property language, 
% then defining an interpretation $\sem{\cdot}$ which maps 
% every formula $\varphi$ in the set of processes that 
% satisfy $\varphi$.\\
% The \textit{Hennessy Milner Logic} \cite{hml} is a property language 
% Which focuses on the actions that can be performed by processes. 
% Informally speaking, given a formula $\varphi$ in this logic, 
% it is possible to build other formulae which express either the capability 
% of a process to reach a state in which $\varphi$ 
% is satisfied, by performing a given action, 
% or the requisite that every state reached after performing 
% a given action satisfies $\varphi$. Different interpretations of formulae of this 
% property language have been investigated, mainly by considering strong or weak actions. 
% Further, recursive variants of the logic have also been proposed \cite{aceto}.

% In this extended abstract we address the question of determining which formulae 
% of a recursive version of the Hennessy Milner Logic, which we 
% will refer to as \rechml, can be used to characterise tests. 

To explain our results, at least intuitively, let us introduce some
informal notation; formal definitions will be given later in the paper.
Suppose we have a property $\phi$ and a test $t$ such that:
\begin{quote}
  for every process $p$,\;\; $p$ satisfies $\phi$ if and only if $p$ \maysatisfy\ 
  the test $t$.
\end{quote}
Then we say the formula $\phi$ \emph{may}-represents the test $t$.
We use similar notation with respect to \emph{must} testing. 
Our first result shows that
the power of tests can be captured by properties; for every test $t$ 
\begin{enumerate}[(i)]
\item 
There is a formula $\phimay{t}$ which \emph{may}-represents $t$; see Theorem~\ref{thm:maytest}
%$p \maysatisfy t$ if and only if  $p$ satisfies the formula 
%$\phimay{T}$

\item 
There is a formula $\phimust{t}$ which \emph{must}-represents $t$; see Theorem~\ref{thm:musttest}
% $p \mustsatisfy t$ if and only if  $p$ satisfies the formula 
% $\phimust{t}$
\end{enumerate}

Properties, or at least those expressed in \rechml, are more
discriminating than tests, and so one would not expect the converse to hold.
But we can give simple descriptions of subsets of \rechml, called 
\mayhml and \musthml respectively, with the following properties:
\begin{enumerate}[(a)]
\item Every $\phi \in \mayhml$ \emph{may}-represents some   test $\Tmay{\phi}$; see Theorem~\ref{thm:mayhml}

\item Every $\phi \in \musthml$ \emph{must}-represents  some  test $\Tmust{\phi}$; see Theorem~\ref{thm:musthml}

\end{enumerate}
Moreover because the formulae $\phimay{t},$\;$\phimust{t}$ given in (i),
(ii) above are in  \mayhml, \musthml respectively, these sub-languages
of \rechml have a pleasing completeness property. For example let
$\phi$ be any formula from \rechml which can be represented by some
test $t$ with respect to \emph{must} testing; that is $p$ satisfies $\phi$ if and only if $p
\mustsatisfy t$. Then, up to logical equivalence, the formula $\phi$ is
guaranteed to be already in the sub-language \musthml; that is, there is
a formula $\psi \in \musthml$ which is logically equivalent to
$\phi$. The language \mayhml has a similar completeness property for
\emph{may} testing.

We now give a brief overview of the remainder of the paper.  In the
next section we recall the formal definitions required to state our
results precisely. Our results in the may case will only hold 
when the set of tests we consider come from a finite state 
finite branching LTS. Further, we also require 
for the LTS of processes to be finite branching when dealing with the 
\must testing relation. The reader should also be warned 
that we use a slightly non-standard interpretation of \rechml.

We then explain
both \emph{may} and \emph{must} testing, where we take as processes
the set of states from an arbitrary LTS, and give an explicit syntax
for tests.  In Section~\ref{sec:tf} we give a precise statement of
our results, including definitions of the sub-languages \mayhml and
\musthml, together with some illuminating examples. The proofs of
these results for the \emph{must} case are given in
Section~\ref{sec:must}, while those for the \emph{may} case are
outlined in Section~\ref{sec:may}. We end with a brief comparison with
related work.

\section{Background}\label{sec:background}
One formal model for describing the behaviour of a concurrent system
is given by \textbf{Labelled Transition Systems (LTSs)}:
% Such a model provides the capability to give a description of concurrent 
% systems by specifying how the states of such system 
% evolve according to a performed action.

\begin{defi}
A LTS over a set of actions $Act$\ is a triple $\mathcal{L} = \langle S,\; Act_{\tau},\; \longrightarrow \rangle$\ where:
\begin{itemize}
\item $S$\ is a countable set of states
\item $Act_{\tau} = Act \cup \{\tau\}$\ is a countable set of actions,
where $\tau$ does not occur in $Act$

\item $\longrightarrow \subseteq S \times Act_\tau \times S$\ is a transition relation.
\end{itemize}
We use $a, b, \cdots$\ to range over the set  of external actions $Act$, and $\alpha,
\beta, \cdots$\ to range over  $Act_\tau$.
The standard notation $s \trans{\alpha} s'$\ will be used in lieu of $(s,\alpha,s') \in \longrightarrow$. 
States of a LTS $\mathcal{L}$\ will also be referred to as (term)
\textit{processes} and ranged over by $s,\,s',p,\;q$\qed. 
\end{defi}

Let us recall some standard notation associated with LTSs. We write
$s \trans{\alpha}$\ if there exists some $s'$\ such that $s
\trans{\alpha} s'$, $s \longrightarrow$\ if there exists $\alpha \in
Act_{\tau}$\ such that $s \trans{\alpha}$, and $s \nottrans{\alpha}$, $s
\nottrans{\;}$\ for their respective negations. We use $\Succ{\alpha, s}$ to
denote the set $\{s' | s \trans{\alpha} s'\}$, and $\Succ{s}$ for $\bigcup_{\alpha \in Act_{\tau}} \Succ{\alpha, s}$. 
If $\Succ{s}$ is finite for every state $s \in S$ the LTS is said to be \textit{finite
  branching}.  Finally, a state $s$\ diverges, denoted $s \Uparrow$,
if there is an infinite path of internal moves $s \trans{\tau} s' \trans{\tau} \cdots$,
while it converges, $s \Downarrow$, otherwise.

For a given LTS, each action of the form $\trans{a}$\ can be interpreted 
as an observable activity; informally speaking, this means that each 
component which is external to the modeled system can detect that such an action 
has been performed. On the other hand, the action $\tau$\ is meant to represent 
internal unobservable activity; this gives rise to the standard notation for 
weak actions. 
$s \Trans{\tau} s'$ Is used to denote reflexive transitive closure of
$\trans{\tau}$, while $s \Trans{a} s'$ denotes $s
\Trans{\tau} s'' \trans{a} s''' \Trans{\tau} s'$.
When $s \Trans{\alpha} s'$\ we say that $s'$ is an 
$\alpha$-derivative of $s$.
The associated notation  $s \Trans{\alpha}$,  
$s \Longrightarrow$, $s \notTrans{\alpha}$\ 
and $s \notTrans{\;}$\ have the obvious definitions.

It is common to define many operators on LTSs for interpreting
process algebras.  In this paper we will use only one, a parallel
operator designed with \emph{testing} in mind.

\begin{defi}[Parallel composition]~\\
Let $\mathcal{L}_1 = \langle S_1,\; Act^{1}_{\tau},\; 
\longrightarrow\rangle$, $\mathcal{L}_2 = \langle S_2,\;
 Act_{\tau}^2,\; \longrightarrow \rangle$ be LTSs.
The parallel composition of $\mathcal{L}_1$\ and $\mathcal{L}_2$\
 is a LTS $\mathcal{L}_1 | \mathcal{L}_2 =\; 
\langle S_1 \times S_2,\; \{\tau\}, \longrightarrow \rangle$, where
$\longrightarrow$\ is defined by the following SOS rules:\\[3pt]

\begin{center}
\begin{tabular}{ccc}
\begin{prooftree}
s \trans{\tau} s'
\justifies
s | t \trans{\tau} s' | t
\end{prooftree}
&
\hspace{50pt}
\begin{prooftree}
t \trans{\tau} t'
\justifies
s | t \trans{\tau} s | t'
\end{prooftree}
\hspace{50pt}
&
\begin{prooftree}
s \trans{a} s' \quad t \trans{a} t'
\justifies
s | t \trans{\tau} s' | t'
\end{prooftree}
\end{tabular}
\end{center}
~\\[7pt]
$s \;|\; t$\ is used as a conventional notation for $(s, t)$.\qed
\end{defi}

The first two rules express the possibility for each component of a LTS to 
perform independently an internal activity, which cannot be detected by the 
other component.
\leaveout{
The first two rules models the possibility for each component of a LTS to
perform their internal actions independently from the other one. This is needed, 
as internal activities of a component cannot be detected by the other one.
}
The last rule models the synchronization of two processes executing the 
same action; this will result in  unobservable activity.
\subsection{Recursive HML}\label{sec:recursivehml}
% One way to analyze the behaviour of a state in a LTS is to formalize
% properties that a state of a LTS satisfies in a modal logic, and then
% by checking if a process satisfies a given property. The
\textbf{Hennessy Milner Logic} (HML), \cite{hml} has proven to 
be a very expressive property language for states in an LTS. 
It is
based on a minimal set of modalities to 
capture the actions a process can perform, and what the effects of performing such 
actions are. Here we use a variant in which the interpretation depends on
the weak actions of an LTS.
\begin{defi}[Syntax of \rechml]
Let $Var$\ be a countable set of variables.
The language \rechml is defined as the set of 
closed formulae generated by the following grammar:
\begin{equation*}
\phi \is \ttt \barra \fff \barra X \barra \Acc A \barra 
\phi_1 \vee \phi_2 \barra \phi_1 \wedge \phi_2 \barra \dmnd{\alpha}\phi \barra [\alpha]\phi \barra 
\lfp{X}{\phi} \barra \gfp{X}{\phi}
\end{equation*}
Here $X$\ is chosen from the countable set of variables $Var$.
The operators $\lfp X{\phi},$\\$\gfp X {\phi}$ act as binders for variables and we have the
standard notions of free and bound variables, and associated binding sensitive
substitution of formulae for variables.\qed
\end{defi}
Let us recall the informal meaning of \rechml operators. A formula 
of the form $\dmnd \alpha \phi$\ expresses the need for a process to have an 
$\alpha$-derivative which satisfies formula $\phi$, while formula $[\alpha]\phi$\ 
expresses the need for all $\alpha$-derivatives (if any) of a converging process to satisfy formula 
$\phi$.\\
Formula $\Acc A$\ is defined when $A$\ is a finite subset of $Act$, and is satisfied exactly 
by those converging processes for which each $\tau$-derivative has at least an $a$-derivative for $a \in Act$.
$\lfp X \phi$\ and $\gfp X \phi$ allow the description of recursive properties, respectively being
the least and largest solution of the equation $X = \phi$\ over the powerset domain of the state space.

Formally, given a LTS $\langle S, Act_{\tau}, \longrightarrow \rangle$, 
we interpret each (closed) formula
as a subset of $2^S$. The set $2^s$ is a complete lattice and the
semantics  is determined by interpreting each operator in the language as 
a monotonic operator over this complete lattice. The binary operators 
$\vee,\;\wedge$ are interpreted as set theoretic union and intersection
respectively while the unary operators are interpreted as follows:
\begin{align*}
  \dmnd{\cdot\alpha\cdot}P = &\;\setof{s}
            { s\Trans{\alpha} s' \mbox{ for some } s' \in P}\\
 \bbox{\cdot\alpha\cdot}P =&\; \setof{s}
        {s\Downarrow, \text{ and } s\Trans{\alpha} s' \mbox{ implies } s' \in P}
\end{align*}
where $P$ ranges over subsets of $2^S$.

Open formulae in \rechml can be interpreted by specifying, 
for each variable $X$, the set of states for which 
the atomic formula $X$\ is satisfied.
Such a mapping $\rho: Var \rightarrow 2^S$\ is called environment.
Let $\env$\ be the set of environments. A formula $\phi$\ of
$\rechml$ will be interpreted as a function $\sem{\phi}: \env \rightarrow 2^S$. 
We will use the standard notation $\rho[X \mapsto P]$\ to refer to the 
environment $\rho'$\ such that $\rho'(X) = P$\ and $\rho'(Y) = \rho(Y)$\ for 
all variables $Y$\ such that $X \neq Y$.\\
The definition of the interpretation $\sem{\cdot}$\ is given in Table
\ref{tab:interpr}.
When referring to the interpretation of a closed formula $\phi \in
\rechml$, we will omit the environment application, and  sometimes 
use the standard notation $p \models \phi$ for $p \in \sem{\phi}$.

\begin{table}[t]%[ht]
\label{tab:interpr}
\begin{equation*}
\begin{array}{lcl|lcl}
\sem{\ttt}\rho &\triangleq& S&
\sem{\fff}\rho &\triangleq& \emptyset\\
\sem{X}\rho &\triangleq& \rho(X)&
\sem{\Acc A}\rho &\triangleq & \{ s | s \Downarrow, \mbox{ if } s \Trans{\tau} s' \mbox{ then } \exists a \in A.s' \Trans{a} \}\\
\sem{\dmnd{\alpha} \phi}\rho &\triangleq& \dmnd{\cdot\alpha\cdot} (\sem{\phi}\rho)&
\sem{\bbox{\alpha} \phi}\rho &\triangleq& \bbox{\cdot\alpha\cdot} (\sem{\phi}\rho)\\
\sem{\phi_1 \vee \phi_2}\rho &\triangleq& \sem{\phi_1}\rho \cup \sem{\phi_2}\rho&
\sem{\phi_1 \wedge \phi_2} \rho &\triangleq& \sem{\phi_1}\rho \cap \sem{\phi_2}\rho\\
\sem{\lfp X \phi}\rho &\triangleq& \bigcap \{ P \;|\; \sem{\phi}\rho[X \mapsto P] \subseteq P\}&
\sem{\gfp X \phi}\rho &\triangleq& \bigcup \{ P \;|\; P \subseteq \sem{\phi}\rho[X \mapsto P]\}\\
\end{array}
\end{equation*}
\caption{Interpretation of \rechml}
\end{table}
Our version of HML is non-standard, as we have added a convergence
requirement for the interpretation of the box operator
$\bbox{\alpha}$.  The intuition here is that, as in the \emph{failures
  model} of CSP \cite{csp}, divergence represents
\emph{underdefinedness}.  So if a process does not converge all of its
capabilities have not yet been determined; therefore one can not quantify over all
of its $\alpha$ derivatives, as the totality of this set has not yet been determined.
Further, the operator $\Acc A$\ is also non-standard. It has been introduced 
for the sake of simplicity, as it will be useful later; in fact it does not add any 
expressive power to the logic, since for each finite set $A \subseteq Act$\ the formula 
$\Acc A$\ is logically equivalent to $[\tau](\bigvee_{a \in A} \dmnd a \ttt)$.

As usual, we will write $\phi\{\psi/X\}$\ to denote the formula $\phi$\ where all 
the free occurrences of the variable $X$\ are replaced with $\psi$. 
We will use the congruence symbol $\equiv$\ for syntactic equivalence.

The language \rechml can be extended conservatively by adding
simultaneous fixpoints, leading to the language $\rechml^+$.  Given a sequence of variables
$(\overline{X})$ of length $n > 0$, and a sequence of formulae
$\overline{\phi}$\ of the same length, we allow the formula $min_i(\overline{X},
\overline{\phi})$ for $1 \leq i \leq n$. This formula
will be interpreted as the $i$-th projection of the simultaneous
fixpoint formula.
\begin{defi}[Interpretation of simultaneous fixpoints]
  Let $\overline{X}$\ and $\overline{\phi}$\ respectively be sequences
  of variables and formulae of length $n$.
\begin{eqnarray*}
  \sem{\lfp {\overline{X}}{\overline{\phi}}}\rho &\triangleq& 
  \bigcap \{ \overline{P} \;|\; \sem{\phi_i}\rho[\overline{X}\mapsto\overline{P}]\subseteq P_i \; \forall 1 \leq i \leq n\}\\
  \sem{\slfp i {\overline{X}}{\overline{\phi}}}\rho &\triangleq& \pi_i(\sem{\lfp{\overline{X}}{\overline{\phi}}}\rho)
\end{eqnarray*}
where $\pi_i$\ is the $i$-th projection operator, and intersection over 
vectors of sets is defined pointwise.\qed
\leaveout{
\[
\langle P_1, \cdots, P_n \rangle \cap \langle Q_1, \cdots, Q_n\rangle = \langle P_1 \cap Q_1, \cdots, P_n \cap Q_n\rangle
\]
}
\end{defi}
Again we will omit the environment application if a formula of the form $\slfp i {\overline{X}}{\overline{\phi}}$ 
is closed, that is the only variables that occur in $\overline{\phi}$ are those in $\overline{X}$. 
Intuitively, an interpretation $\sem{\lfp {\overline{X}} {\overline{\phi}}}$, where 
$\overline{X} = \langle X_1,\cdots,X_n\rangle$ and 
$\overline{\phi} = \langle \phi_1, \cdots, \phi_n \rangle$, is the least solution 
(over the set of vectors of length $n$ over $2^S$) of the equation system given 
by $X_i = \phi_i$ for all $i = 1, \cdots, n$, 
\leaveout{
whose form is
\begin{eqnarray*}
X_1 &=& \phi_1\\
&\vdots&\\
X_n &=& \phi_n
\end{eqnarray*}
}
while $\sem{\slfp i {\overline{X}} {\overline{\phi}}}$\ is the $i$-th projection of such a vector.
\leaveout{Let $\overline{P} = \langle P_1, \cdots, P_n\rangle$\ be the least solution for a system of 
equations as above. The following theorem states that, for each index $i$, there exists an 
equation $X = \psi$\ such that its least solution coincides with $P_i$.
}
Simultaneous fixpoints do not add any expressivity to \rechml, as shown below:
\begin{thm}[Bek\'ic, \cite{becik}]
\label{thm:becik}\qquad\\
For each formula $\phi \in \rechml^+$\ there is a formula $\psi \in \rechml$\ such that $\sem \phi = \sem \psi$.\qed
\end{thm}
Later we will need the following properties of simultaneous fixpoints:
\begin{thm}[Fixpoint properties]\qquad
\label{thm:fixpointprop}
  \begin{enumerate}[(i)]
  \item 
\label{thm:minfixprop}
Let $(\overline{P})$\ be a vector of sets from $2^S$\ satisfying
$
\sem{\phi_i} \rho[\overline{X} \mapsto \overline{P}] \subseteq P_i
$ for every $1 \leq i \leq n$.
Then 
$
\sem{min_i(\overline{X}, \overline{\phi})} \rho \subseteq P_i
$
\item
\label{thm:fixprop}
Let $\rhomin$\ be an environments such that
$
\rhomin(X_i) = \sem{min_i(\overline{X}, \overline{\phi})}.
$
Then 
$
\sem{min_i(\overline{X}, \overline{\phi})} = \sem{\phi_i} \rhomin.
$\qed
 \end{enumerate}
\end{thm}

\subsection{Tests}
Another way to analyse the behaviour of a process is given by testing.
Testing a process can be thought of as an experiment in which another
process, called test, detects the actions performed by the tested
process, reacting to it by allowing or forbidding the execution of a
subset of observables. After observing the behaviour of the process,
the test could decree that it satisfies some property for which the
test was designed for by reporting the success of the experiment,
through the execution of a special action $\omega$.

Formally speaking, a  test is a state from a LTS 
$\mathcal{T} = \langle T, Act^\omega_{\tau}, \longrightarrow \rangle$, 
where $ Act^\omega_{\tau} = Act_{\tau} \cup \{\omega\}$ and $\omega$ 
is an action not contained in $Act_{\tau}$. 

Given a LTS of processes $\mathcal{L} = \langle S, Act_{\tau}, \longrightarrow \rangle$, an experiment
consists of a pair $p \;|\; t$ from the  product LTS 
$(\mathcal{L}\barra\mathcal{T})$. We refer to a maximal path
$  p \;|\; t \trans{\tau} p_1 \;|\; t_1 \trans{\tau} \ldots \ldots 
     \trans{\tau} p_k \;|\; t_k \trans{\tau} \ldots
$
as a \emph{computation} of $p \;|\; t$. It may be finite or infinite; it is successful if there 
exists some $n \geq 0$ such that $t_n \trans{\omega}$. As only $\tau$-actions can be performed in 
an experiment, we will omit the symbol $\tau$ in computations and in computation prefixes.
Successful computations lead to the definition of two well known \textit{testing relations}, \cite{dhn}:
\begin{defi}[May Satisfy, Must Satisfy] Assuming a LTS of processes and a LTS of tests, 
let $s$  and $t$ be a state and a test from such LTSs, respectively. We say
\begin{enumerate}[(a)]
\item $s \maysatisfy t$ if there exists a successful computation for the experiment 
$s \;|\; t$.
\item $s \mustsatisfy t$ if each computation of the experiment $s\; |\; t$\ is successful.
\end{enumerate}
\end{defi}
\leaveout{
Processes can now be compared in terms of the set of test that they may/must pass.
}
Later in the paper we will use a specific LTS of tests, whose states are all
the closed  terms generated by the grammar
\begin{equation}
t \is 0 \barra \alpha.t \barra \omega.0 \barra  X \barra t_1 + t_2 \barra \mu X.t \; .
\label{eq:tests}
\end{equation}
Again in this language  $X$\ is bound in $\mu X.t$, and  the test $t\{t'/X\}$ 
denotes the test $t$ in which  each free occurrence of $X$ is replaced by $t'$.
The transition relation is  defined by the following rules:\footnote{For the sake of clarity, 
the rules use an abuse of notation, by considering $\alpha$\ as an action from $Act_{\tau} \cup {\omega}$\ rather than from $Act_{\tau}$.}
%\begin{center}
\begin{displaymath}
\begin{tabular}{llll}
\begin{prooftree}
\;
\justifies
\alpha.t \trans{\alpha} t
\end{prooftree}
&\qquad
\begin{prooftree}
t_1 \trans{\alpha} t_1'
\justifies
t_1 + t_2 \trans{\alpha} t_1'
\end{prooftree}
&\qquad
\begin{prooftree}
t_2 \trans{\alpha} t_2'
\justifies
t_1 + t_2 \trans{\alpha} t_2'
\end{prooftree}
%\\&&\\
&\qquad
\begin{prooftree}
\;
\justifies
\mu X.T \trans{\tau} t\{(\mu X.t)/X\}
\end{prooftree} %&\\&&
\end{tabular}
%\end{center}
\end{displaymath}

The last rule states that a test of the form $\mu X.t$\ can always perform a $\tau$-action before evolving in the test $t\{\mu X.t/X\}$. 
This treatment of recursive processes will allow us to prove properties 
of paths of recursive tests and experiments by performing an induction on their length. 
Further, the following properties hold for a test $t$ in grammar \eqref{eq:tests}:

\begin{prop}
\label{prop:Tbf}

Let $\mathcal{T} = \langle T, Act_{\tau}, \longrightarrow \rangle$\ be the LTS generated by a state $t$\ in grammar \eqref{eq:tests}: then 
$\mathcal{T}$ is both branching finite and finite state.\qed
\end{prop}

\section{Testing formulae}\label{sec:tf}
Relative to a  process LTS $\langle S, Act_{\tau}, \longrightarrow \rangle$\ 
and a test LTS $\langle T, Act_{\tau}^{\omega}, \longrightarrow \rangle$, we 
now explore the relationship between tests from our
default LTS of tests and  formulae of \rechml. 
Given a test $t$, our goal is to find a formula $\phi$\ such that 
the set of processes which \maysatisfy/\mustsatisfy 
such a test is completely characterised by the 
interpretation $\sem{\phi}$. Moreover, we 
aim to establish exactly the subsets of 
\rechml for which each formula can be checked by 
some test, both in the \may and \must case.
 
For this purpose some definitions are necessary:
\begin{defi}
Let $\phi$\ be a \rechml formula and $t$ a test. We say that:
\begin{itemize}
\item $\phi$ \emph{must}-represents  the test $t$,  if
for all  $p \in S$,  $p \mustsatisfy t$  
if and only if $p \models \phi$.

\item $\phi$ is  \emph{must}-testable whenever there exists 
a test  which  $\phi$ \emph{must}-represents.

\item $t$ is \emph{must}-representable, if
  there exists some $\phi \in \rechml$ which
  \emph{must}-represents it respectively.
\end{itemize}
Similar definitions are given for may testing.\qed
\end{defi}
First some examples.
\begin{example}[Negative results]
\begin{enumerate}[(a)]\qquad
\item $\phi = [a]\fff$ is not \emph{may}-testable.\\
Let $s \in \sem{[a]\fff}$; a new process $p$\ can be built starting from $s$\ by letting  
$p \trans{\tau} p$, whenever
$s \trans{\alpha} s'$ then $p \trans{\alpha} s'$.\\
Processes $p$\ and $s$\ \maysatisfy the same set of tests. However, $p \notin \sem{[a]\fff}$, as $p \Uparrow$. 
Therefore\\ no test \emph{may}-represents $[a]\fff$.
\item $\phi = \dmnd a \ttt$\ is not \emph{must}-testable.\\
We show by contradiction that there exists no test $t$ that \emph{must}-represents $\phi$. 
To this end, we perform a case analysis on the structure of $t$.
\begin{itemize}
\item $t \trans{\omega}$: Consider the process $0$ with no transitions. Then $0 \mustsatisfy t$, 
whereas $0 \notin \sem\phi$.
\item $t \nottrans{\omega}$: Let $s \in \sem\phi$ and consider the process $p$ built up 
from $s$ according to the rules of the example above; we have $p \in \sem\phi$. On the 
other hand, $p \mustsatisfy t$ is not true; indeed the experiment $p \;|\;t$ leads to 
the unsuccessful computation 
$p\;|\;t \shortrightarrow p\;|\;t \shortrightarrow \cdots.$
\end{itemize}
Therefore there is no test $t$ which \emph{must}-represents $\phi$.
\item $\phi = \dmnd a \ttt \wedge \dmnd b \ttt$\ is not \emph{may}-testable.\\
Let $s$\ be the process whose only transitions are $s \trans{a} 0$, $s \trans{b} 0$.
Let also $p, p'$ be the processes whose only transitions are 
$p \trans{a} 0$, $p' \trans{b} 0$. We have $s \in \sem{\phi}$, whereas 
$p, p' \notin \sem{\phi}$. We show that whenever $s \maysatisfy$ a test $t$, 
then either $p \maysatisfy t$ or $p' \maysatisfy t$. Thus there exists 
no test which is \emph{may}-satisfied by exactly those processes in 
$\sem{\phi}$, and therefore $\phi$ is not 
\emph{may}-representable.
\leaveout{Let $t$ be an 
arbitrary test such that $s \maysatisfy t$; we will 
show that either $p \maysatisfy t$\ or $p' \maysatisfy t$ 
hold, where $p$\ and $p'$\ are the processes whose 
only transitions are $p \trans{a} 0, p' \trans{b} 0$.\\
As neither $p$\ and $p'$\ satisfy the formula $\phi$, 
it follows that there exists a process $q$\ 
such that $q \maysatisfy t$\ but 
$q \notin \sem{\phi}$, hence 
$\phi$\ is not \emph{may}-testable.\\
}
First, notice that if $s \maysatisfy t$, then at least one of the following holds:
\begin{enumerate}[(i)]
\item $t\Trans{\omega}$, \label{cond:1}
\item $t\Trans{a}t'\Trans{\omega}$, \label{cond:2}
\item $t\Trans{b}t'\Trans{\omega}$. \label{cond:3}
\end{enumerate}

If $t\Trans{\omega}$, then trivially both $p$\ and $p'$\ \maysatisfy $t$. On the other hand, if $t\Trans{a}t'\Trans{\omega}$, 
then there exist $t'', t_{\omega}$\ such that $t \Trans{\tau} t'' \trans{a} t' \Trans{\tau} t_{\omega} \trans{\omega}$. We can 
build the computation fragment for $p \barra t$\ such that
\[
p \barra t \shortrightarrow \cdots \shortrightarrow p \barra t'' \shortrightarrow 0 \barra t' \shortrightarrow \cdots \shortrightarrow 0 \barra t_{\omega}
\]

which is successful. Hence $p \maysatisfy t$. Finally, The case $t\Trans{b}t'\Trans{\omega}$ is similar.
\label{ex:c}
\item In an analogous way to \eqref{ex:c} it can be shown that $[a] \fff \vee [b] \fff$\ is not \emph{must}-testable.\qed
\end{enumerate}
\end{example}
We now 
investigate precisely which  formulae in \rechml can be represented by tests. 
To this end, we define two sub-languages, namely \mayhml and \musthml.
\begin{defi}{(Representable formulae)}
\begin{itemize}
\item The language \mayhml is defined to be the set of closed formulae generated by the following \rechml grammar fragment:
\begin{eqnarray}
\phi \is \ttt \barra \fff \barra X \barra \dmnd \alpha \phi \barra \phi_1 \vee \phi_2 \barra \lfp X \phi
\end{eqnarray}
\item The language \musthml is defined to be the set of closed formulae generated by the following \rechml grammar fragment:
\begin{eqnarray}
\phi \is \ttt \barra \fff \barra \Acc A \barra X \barra [\alpha]\phi \barra \phi_1 \wedge \phi_2 \barra \lfp X \phi 
\end{eqnarray}
\end{itemize}
\end{defi}

\noindent
Note that both sub-languages use the minimal fixpoint operator only;
this is not surprising, as informally at least testing is an inductive
rather than a coinductive property.  Since there exist formulae of the
form $[\alpha]\phi$, $\phi_1 \wedge \phi_2$ which are not
\emph{may}-representable, the $[\cdot]$ modality and the conjunction
operator, have not been included in \mayhml \leaveout{ The modality
  $[\cdot]$\ and the conjunction operator $\wedge$ are not allowed in
  \mayhml; the above examples show in fact that there exist formulae
  of the form $[\alpha]\phi$\ which are not \emph{may}-testable, and
  that conjunction of two formulae is not always \emph{may}-testable.
} The same argument applies to the modality $\dmnd \cdot$ and the
disjunction operator $\vee$\ in the must case, which are therefore
not included in \musthml.

Note also that the modality $\bbox{\cdot}$ is only used in \musthml,
which will be compared with \emph{must-testing}. No diverging process 
must satisfy a non-trivial test $t$, i.e. such that  $t \nottrans{\omega}$.
Hence, in this setting, the convergence restriction  on
this modality is natural. 

% Finally, we can now explain why we added a convergence requirement to
% the modality $\bbox{\cdot}$. Such a modality is used only in \musthml,
% which will be compared with \emph{must-testing}. It is important to
% notice that no diverging process \mustsatisfy a non trivial test $t$,
% id est such that $t \nottrans{\omega}$. Thus, for a formula $\phi$ to
% \emph{must}-represent a non trivial test $t$, it is mandatory to
% impose a convergence requirement.

We have now completed the set of definitions setting up our framework of 
properties and tests. In the remainder of the paper we prove the results 
announced, informally, in the Introduction.

\section{The must case}\label{sec:must}

We will now develop the mathematical basis needed to relate \musthml
formulae and the \must testing relation; in this section we will assume
that the LTS of processes is branching finite.\\

\begin{lem}
\label{lem:divergence}
Let $\phi \in \musthml$, and let $p \in \sem{\phi}$, where $p
\Uparrow$: then $\sem{\phi}$\ is the entire process space,
i.e. $\sem{\phi} = S$. \qed
\end{lem}
\noindent
This lemma has important consequences; it means formulae in \musthml either have the trivial interpretation as
the full set of states $S$, or they are only satisfied by convergent states. 
\begin{defi}
  Let $\mathcal{C}$ be the set of subsets of $S$ determined by:
\begin{itemize}
\item $S \in \mathcal{C}$,
\item $X \in \mathcal{C}, s \in X$ implies $s \Downarrow$. \qed
\end{itemize}
\end{defi}

\begin{prop}
\label{prop:cpo}
$\mathcal{C}$ ordered by set inclusion is a  \emph{continuous partial
order}, \emph{cpo}.
\end{prop}
\begin{proof}
  The empty set is obviously the least element in $\mathcal{C}$. So it is sufficient to show
that if  $X_0 \subseteq X_1 \subseteq \cdots$\ is a chain of elements  in $\mathcal{C}$ then 
$\bigcup_n X_n$ is also in $\mathcal{C}$. 
\end{proof}
We can now take advantage of the fact that \musthml actually has 
a continuous interpretation in $(\mathcal{C}, \subseteq)$. 
The only non trivial case here is the continuity of 
the operator $\bbox{\cdot\alpha\cdot}$:
\begin{prop}
\label{prop:dmndcontinuous}
Suppose the LTS of processes is finite-branching: If  
$X_0 \subseteq X_1 \subseteq \cdots$\ is a chain 
of elements in $\mathcal{C}$ then
\[
\bigcup_n [\cdot \alpha \cdot] X_n = [\cdot \alpha \cdot] \bigcup_n X_n.
\]
\qed
\end{prop}
This continuous interpretation of \musthml allows us to 
use  chains of finite
approximations for these formulae of \musthml. 
That is given $\phi \in \musthml$\ and $k\geq 0$, recursion free
formulae $\phi^k$\ will be defined such that $\sem{\phi^k} \subseteq
\sem{\phi^{(k+1)}}$\ and $\bigcup_{k\geq 0} = \sem{\phi}$. We can therefore
 reason inductively on approximations in order to prove
properties of recursive formulae.
\begin{defi}[Formulae approximations]
For each formula $\phi$\ in \musthml define
\begin{eqnarray*}
\phi^0 &\triangleq& \fff\\
\phi^{(k+1)} &\triangleq& \phi \mbox{\hspace{135pt}if } \phi = \ttt,\fff \mbox{ or } \Acc A\\
\leaveout{
\ttt^{(k+1)} &\triangle eq& \ttt\\
\fff^{(k+1)} &\triangleq &\fff\\
(\Acc A)^{(k+1)} &\triangleq& \Acc A\\
}
([\alpha]\phi)^{(k+1)} &\triangleq& [\alpha](\phi)^{(k+1)}\\
(\phi_1 \wedge \phi_2)^{(k+1)} &\triangleq& \phi_1^{(k+1)} \wedge \phi_2^{(k+1)}\\
(\lfp X \phi)^{(k+1)} &\triangleq& (\phi\{min(X, \phi)/X\})^k
\end{eqnarray*}\qed
\end{defi}

It is obvious that for every $\phi \in \musthml$, $\sem{\phi^k} \subseteq
\sem{\phi^{(k+1)}}$ for every $k \geq 0$; The fact that the union of the
approximations of $\phi$\ converges to $\phi$\ itself depends on the
continuity of the interpretation: 
\begin{prop}
\label{cor:continuity}
\[
\bigcup_{k\geq 0} \sem{\phi^k} = \sem{\phi}
\]
\end{prop}
\begin{proof}
 This is true in the initial continuous interpretation of the language, and therefore also in our interpretation.
  For details see  \cite{finiteapprox}. 
\end{proof}

Having established these properties of the interpretation of formulae
in \musthml, we now show that they are all \emph{must}-testable.  The
required tests are defined by induction on the structure of the
formulae.
\begin{defi}
For each $\phi$\ in \musthml define $\Tmust \phi$\ as follows:
\begin{eqnarray}
\Tmust \ttt &=& \omega.0 \label{eq:tmusttt}\\
\Tmust \fff &=& 0 \label{eq:tmustff}\\
\Tmust {\Acc A} &=& \sum_{a \in A} a.\omega.0 \label{eq:tmustacc}\\
\Tmust {X} &=& X \label{eq:tmustX}\\
\Tmust {[\tau] \phi} &=& \tau.\Tmust \phi \label{eq:tmusttau}\\
\Tmust {[a] \phi} &=& a. \Tmust \phi + \tau.\omega.0 \label{eq:tmusta}\\
%\Tmust {\phi_1 \wedge \phi_2} &=& \omega.0 \;\;\;\mbox{if } \phi_1 \wedge \phi_2\ \mbox{ is closed and logically equivalent to }\ttt\\
%\Tmust {\phi_1 \wedge \phi_2} &=& \tau.\Tmust {\phi_1} + \tau. \Tmust{\phi_2}\;\;\;\mbox{otherwise}\\
\Tmust {\phi_1 \wedge \phi_2} &=& \begin{cases} 
	\omega.0 & \mbox{if } \phi_1 \wedge \phi_2 \mbox{ is closed and }\\
	&\mbox{logically equivalent to }\ttt\\
	&\\
	\tau.Tmust {\phi_1} + \tau. \Tmust{\phi_2}&\mbox{otherwise}
	\end{cases} \label{eq:tmustwedge}\\
%\Tmust {\lfp X \phi} &=& \Tmust \phi \;\;\;\mbox{ if } \phi \mbox{ is closed}\\
%\Tmust {\lfp X \phi} &=& \mu X. \Tmust \phi \;\;\; \mbox{otherwise}
\Tmust{\lfp X \phi} &=& \begin{cases}
	\Tmust \phi & \mbox{ if } \phi \mbox{ is closed}\\
	\mu X.\Tmust \phi & \mbox{otherwise}
	\end{cases} \label{eq:tmustmin}
\end{eqnarray}\qed
\end{defi}

For each formula $\phi$\ in $\musthml$, the test $\Tmust \phi$\ is defined 
in a way such that the set of processes which  $\mustsatisfy$\ $\Tmust\phi$\ is exactly $\sem{\phi}$. 
Before supplying the details of a formal proof of this statement, let us comment on the definition of $\Tmust\phi$.\\
Cases (\ref{eq:tmusttt}), (\ref{eq:tmustff}) and (\ref{eq:tmustX}) are straightforward.
In the case of $\Acc A$, the test allows only those action which are in $A$\ to be performed by a process, after which it reports success.\\
For the box operator, a distinction has to be made between $[a]\phi$\ and $[\tau]\phi$. In the former we have 
to take into account that a converging process which cannot perform a weak $a$-action 
satisfies such a property; thus, synchronisation through the execution of a $a$-action is allowed, but a possibility for the test to report success 
after the execution of an internal action is given.
In the case of $\bbox{\tau}\phi$ no synchronization with any action is required; however, 
since we are adding a convergence requirement to formula $\phi$, we have to avoid the possibility that the 
test $\Tmust{\bbox{\tau}\phi}$ can immediately perform a $\omega$ action. This is done by requiring the test 
$\Tmust{\bbox{\tau}\phi}$ to perform only an internal action.\\
Finally, (\ref{eq:tmustwedge})\ and (\ref{eq:tmustmin}) are defined by distinguishing between two cases; this is because a formula of the form $\phi_1 \wedge \phi_2$\ or $\lfp X \phi$\ can be logically equivalent to $\ttt$, whose interpretation is the entire state space. However, the second clause in the definition of $\Tmust \phi$\ for such formulae 
require the test to perform a $\tau$\ action before performing any other activity, thus at most converging processes \mustsatisfy such a test.

In order to give a formal proof that $\Tmust{\phi}$ does indeed capture the formula $\phi$ we need to establish some preliminary
properties. The first essentially says that no formula of the form $\lfp X \phi$, with $\phi$\ not closed, will be interpreted in the whole state space.

\begin{lem} 
\label{lem:statespaceformulae}
Let $\phi = \lfp X \psi$, with $\psi$\ not closed. Then $\sem{\phi} \neq S$.\qed
\end{lem}

Then we state some simple properties about recursive tests.

\begin{lem}\qquad
\label{lem:testprops}
\begin{itemize}
\item $p \mustsatisfy \mu X.t$\ implies $p \mustsatisfy t\{\mu X.t/X\}$.
\item $p \Downarrow, p \mustsatisfy t\{\mu X.t/X\}$\ implies $p \mustsatisfy \mu X.t$.\qed
\end{itemize}
\end{lem}
Note that the premise $p \Downarrow$  is essential in the second part of this lemma, 
as $\mu X.t$ cannot perform a $\omega$ action; therefore it can be \emph{must}-satisfied 
only by processes which converge.

%These results allow us to estabish one implication of theorem \ref{thm:musthml}.
\begin{prop}\label{prop:oneway}
Suppose the LTS of processes is finitely branching. If $p \mustsatisfy \Tmust \phi$\ then $p \in \sem{\phi}$.
\end{prop}

\begin{bproof} Suppose $p \mustsatisfy \Tmust \phi$; As both the LTS of
  processes (by assumption) and the LTS of tests (Proposition
  \ref{prop:Tbf}) are finite branching, the
  maximal length of a successful computation $|p,t|$\ is defined and
  finite.  This is a direct consequence of Konig's Lemma
  \cite{Boolos}.  Thus it is possible to perform an induction over
  $|p,\Tmust \phi|$ to prove that $p \in \sem{\phi^k}$.  The result
  will then follow from Proposition \ref{cor:continuity}.
\begin{itemize}
\item If $|p, \Tmust \phi| = 0$\ then $\Tmust \phi \trans{\omega}$, and hence for each $p \in S\; p \mustsatisfy\ \Tmust \phi$. Further it is not difficult to show that $\phi$\ is logically equivalent to $\ttt$, hence $p \in \sem{\phi}$.
\item If $|p, \Tmust \phi| = n+1$\ then the validity of the Theorem follows from an application of an inner induction on $\phi$. We show only the most interesting case, which is $\phi = \lfp X \psi$. There are two possible cases.
\begin{enumerate}[(a)]
\item If $X$\ is not free in $\psi$\ then the result follows by the inner induction, as $\lfp X \psi$\ is logically equivalent to $\psi$, and $\Tmust{\lfp X \psi} \equiv \Tmust \psi$\ by definition.
\item If $X$\ is free in $\psi$\ then, by Lemma \ref{lem:testprops}\ $p \mustsatisfy \Tmust \psi \{\mu X.\Tmust{\psi}/X\}$, which is syntactically equal to $\Tmust {\psi\{ \lfp X \psi / X\}}$.\\
Since $|p, \Tmust {\psi\{\lfp X \psi / X\}}| < |p, \Tmust \phi|$, by inductive hypothesis we have \\$p \in \sem{\psi \{\lfp X \psi /X\}^k}$\ for some $k$, hence $p \in \sem{\phi^{(k+1)}}$.\qed
\end{enumerate}
\end{itemize}
\end{bproof}

To prove the converse of Proposition~\ref{prop:oneway} we use the following concept:
\begin{defi}[Satisfaction Relation]
Let $R \subseteq S \times \musthml$ and for any $\phi$\ let
$(R\; \phi) = \{ s \barra s\; R\;\phi\}$.\\
Then $R$\ is a satisfaction relation if it satisfies
\begin{eqnarray*}
(R\;\ttt) &=& S\\
(R\;\fff) &=&\emptyset\\
 (R\; \Acc A) &=&\setof{s}{s \Downarrow, s \Trans{\tau} s' \mbox{ implies } S(s') \cap A \neq \emptyset}\\
(R \;[\alpha]\phi) &\subseteq& [\cdot \alpha \cdot] (R\; \phi)\\
(R\; \phi_1 \wedge \phi_2) &\subseteq& (R\;\phi_1) \cap (R\; \phi_2)\\
(R \; \phi\{\lfp X \phi /X\}) &\subseteq& (R\; \lfp X \phi)
\end{eqnarray*}\qed
\end{defi}

Satisfaction relations are defined to agree with the interpretation $\sem \cdot$. Indeed, all implications required for 
satisfaction relations are satisfied by $\models$. Further, as $\sem{\lfp X \phi}$\ is defined to be the least solution to the recursive 
equation $X = \phi$, we expect it to be the smallest satisfaction relation.

\begin{prop}\label{prop:satisfaction}
The relation $\models$\ is a satisfaction relation. Further, it is the smallest satisfaction relation.\qed
\end{prop}
\leaveout{
This Proposition can be exploited to prove properties for couples $(p, \phi)$ such that $p \models \phi$, 
for $\phi \in \musthml$.\\
Let $\pi$\ be a property over $S \times \musthml$, and suppose the relation $R = \{(s, \phi)\barra \pi_(s,\phi)\}$ is a satisfaction relation. 
We obtain, by Proposition \ref{prop:satisfaction}, that $p \models \phi$ implies $\pi(p, \phi)$.\\
}
Proposition \ref{prop:satisfaction} ensures that, 
for any satisfaction relation $R$, $\models$ is 
included in $R$; in other words, if  
$p \models \phi$ then $p \;R\;\phi$.
Next we consider the relation $R_{\scriptstyle{must}}$\ such that $p\; R_{\scriptstyle{must}}\;\phi$ whenever $p \mustsatisfy \Tmust \phi$, 
and show that it is a satisfaction relation.

\begin{prop}\label{prop:must.satisfaction}
The relation $R_{\text{must}}$  is a satisfaction relation.
\end{prop}

\begin{bproof}
We proceed by induction on formula $\phi$. Again, we only check the most interesting case.\\
Suppose $\phi = \lfp X \psi$. We have to show $p \mustsatisfy \Tmust {\psi\{\phi/X\}}$\ 
implies $p \mustsatisfy \Tmust \phi$.\\
We distinguish two cases:
\begin{enumerate}[(a)]
\item $X$\ does not appear free in $\psi$. then 
$\Tmust \phi = \Tmust \psi$, and $\psi\{\phi/X\} = \psi$. This case is trivial.
\item $X$\ does appear free in $\phi$: in this case $\Tmust \phi = \mu X.\Tmust \psi$, 
and $\Tmust {\psi\{\phi/X\}}$\ has the form\\
$\Tmust \psi \{\mu X.\Tmust \psi /X\}$. 
By Lemma \ref{lem:statespaceformulae} $\sem \phi \neq S$; therefore Lemma \ref{lem:divergence}\ 
ensures that $p \Downarrow$, and hence by Lemma \ref{lem:testprops}\ it follows 
$p \mustsatisfy \Tmust \phi$.\qed
\end{enumerate}
\end{bproof}
\noindent
Combining all these results we now obtain our result on the testability of \musthml.
\begin{thm}
\label{thm:musthml}
Suppose the LTS of processes is finite-branching. Then for every 
$\phi \in \musthml$, there exists a   test $\Tmust \phi$\ such that $\phi$ \emph{must}-represents the test $\Tmust \phi$.
\end{thm}
\begin{proof}
We have to show that for any process $p$, $p \mustsatisfy \Tmust \phi$ if and only if $p \in \sem{\phi}$.
  One direction follows from Proposition~\ref{prop:oneway}.  Conversely suppose $p \in \sem{\phi}$. 
By Proposition \ref{prop:satisfaction}\ it follows that for all satisfaction relations 
$R$ it holds $p\; R\; \phi$; hence, by Proposition \ref{prop:must.satisfaction}, $p \;R_{\text{must}}\;\phi$, 
or equivalently $p \mustsatisfy \Tmust \phi$.
\end{proof}

We now turn our attention to the second result, namely that every test $t$ is \emph{must}-representable by some formula 
in \musthml. Let us for the moment assume a branching finite LTS of tests in which the state space $T$ is finite.
\begin{defi}\label{def:tests}
Assume we have a test-indexed set of variables $\{X_t\}$.
For each test $t \in T$\ define $\varphi_t$\ as below:
\begin{eqnarray}
\label{eq:must1}
\varphi_t &\triangleq&\ttt \hspace{139pt} \mbox{if } t\trans{\omega}\\
\label{eq:must2}
\varphi_t &\triangleq&\fff \hspace{139pt} \mbox{if } t \nottrans{\;}\\
\label{eq:must3}
\varphi_t &\triangleq& (\displaystyle{
\bigwedge_{\scriptstyle{a,t': t \trans{a} t'}}} [a] X_{t'})
 \;\wedge\; \Acc{\{a | t \trans{a}\}}
 \hspace{20pt} \mbox{if } t \nottrans{\omega}, t \nottrans{\tau}, 
 t \longrightarrow\\
\label{eq:must4}
\varphi_t &\triangleq& (\displaystyle{
\bigwedge_{t': t \trans{\tau}t'}} [\tau]X_{t'})
 \;\wedge\; (\displaystyle{\bigwedge_{a,t': t \trans{a} t'}} 
 [a] X_{t'}) \hspace{20pt}\mbox{if } t \nottrans{\omega}, t\trans{\tau}
\end{eqnarray}

\leaveout{
\begin{displaymath}
\begin{array}{lclcr}
\varphi_t &\triangleq&\ttt & \mbox{if}& t\trans{\omega}\\
\varphi_t &\triangleq&\fff & \mbox{if}& t \nottrans{\;}\\
\varphi_t &\triangleq& (\displaystyle{\bigwedge_{\scriptstyle{a,t': t \trans{a} t'}}} [a] X_{t'}) \;\wedge\; \Acc{\{a | t \trans{a}\}}& \mbox{if} & t \nottrans{\omega}, t \nottrans{\tau}, t \longrightarrow\\
\varphi_t &\triangleq& (\displaystyle{\bigwedge_{t': t \trans{\tau}t'}} [\tau]X_{t'}) \;\wedge\; (\displaystyle{\bigwedge_{a,t': t \trans{a} t'}} [a] X_{t'}) &\mbox{if}& t \nottrans{\omega}, t\trans{\tau}
\end{array}
\end{displaymath}
}
Take $\phi_t$\ to be the extended formula $\slfp t {\overline{X_T}} {\overline{\varphi_T}}$, using  the simultaneous least fixed points
introduced in Section~\ref{sec:recursivehml}.\qed
\end{defi}

Notice that we have a finite set of variables $\{X_t\}$ and that the conjunctions in Definition \ref{def:tests} are finite, as the LTS of tests 
is finite state and finite branching. These two conditions are needed  for $\phi_t$ to be well defined.

Formula $\phi_t$\ captures the properties required by a process to \mustsatisfy\ test $t$. The first two clauses of 
the definition are straightforward. If $t$\ cannot make an internal action or cannot report a success, but can perform a visible action $a$ to 
evolve in $t'$, then a process should be able to perform a $\Trans{a}$\ transition and evolve in a process $p'$ such that $p' \mustsatisfy t'$. 
The requirement $\Acc{\{a \barra t\trans{a}\}}$\ is needed because a synchronisation between the process $p$\ and the test $t$\ is required 
for $p \mustsatisfy t$\ to be true.\\
In the last clause, the test $t$\ is able to perform at least a $\tau$-action. In this case there is no need for a synchronisation 
between a process and the test, so there is no term of the form $\Acc{\{a \barra t\trans{a}\}}$\ in the definition of $\phi_t$. 
However, it is possible that a process $p$\ will never synchronise with such test, instead $t$\ will perform a transition $t \trans{\tau}t'$\ after $p$\ has executed an arbitrary number of 
internal actions. Thus, we require that for each transition $p \Trans{\tau} p'$, $p' \mustsatisfy t'$.

We now supply the formal details which lead to state that formula $\phi_t$\ characterises the test $t$. 
Our immediate aim is to show that the two environments, defined by
\begin{alignat*}{2}
  \rhomin(X_t) &\triangleq \sem{ \phi_t} &\qquad\qquad\qquad
  \rhomust(X_t) &\triangleq \{ p \barra p \mustsatisfy t\}
\end{alignat*}
% \begin{eqnarray*}
% \rhomin(X_t) &\triangleq& \sem{ \phi_t}\\
% \rhomust(X_t) &\triangleq& \{ p \barra p \mustsatisfy t\}
% \end{eqnarray*}
are identical. This is achieved in the following two propositions. 
\begin{prop}
\label{thm:minsubsetmust}
For all $t \in T$ it holds that $\rhomin(X_t) \subseteq \rhomust(X_t)$.
\end{prop}

\begin{proof}
We just need to show that 
$\sem{\varphi_t} \rhomust \subseteq \rhomust(X_t)$: 
the result follows from an application of the 
\textit{minimal fixpoint property}, 
Theorem \ref{thm:fixpointprop} (\ref{thm:minfixprop}). 
The proof is carried out by performing a case 
analysis on $t$. We will only consider Case \eqref{eq:must3}, 
as cases \eqref{eq:must1} and \eqref{eq:must2} 
are trivial and Case \eqref{eq:must4} is handled similarly.

Assume $p \in \sem{\varphi_t}\rhomust$. We have
\begin{enumerate}[(a)]
\item \label{prf:cond1}$p \Downarrow$,
\item \label{prf:cond2}whenever $p \Trans{\tau} p'$\ there exists an action $a \in Act$\ such that $t \trans{a}$\ and $p' \Trans{a}$,
\item \label{prf:cond3}whenever $p \Trans{a} p'$\ and $t \trans{a} t'$, $p' \in \rhomust(X_{t'})$, i.e. $p' \mustsatisfy t'$.
\end{enumerate}

Conditions (\ref{prf:cond1}) and (\ref{prf:cond2}) are met since $p \in \sem{\Acc{\{a \;|\; t \trans{a}}}$ and $t \trans{a}$\ for some $a \in Act$, while (\ref{prf:cond3}) is true because of $p \in \sem{\bigwedge_{a, t':\; t \trans{a}t'}[a]X_{t'}}$.\\\\
To prove that $p \in \rhomust(X_t)$\ we have to show that every computation of $p\;|\;t$\ is successful. To this end, consider an arbitrary computation of $p\;|\; t$; condition (\ref{prf:cond2}) ensures that such a computation cannot have the finite form
\begin{equation}
\label{eq:nonmaximalcomp}
p \barra t \shortrightarrow p_1 \barra t \shortrightarrow \cdots p_k \barra t \shortrightarrow p_{k+1} \barra t \shortrightarrow \cdots \shortrightarrow p_n \barra t
\end{equation}

For such a computation we have that $p_n \Trans{\tau} p'$, and there exists $p''$\ with $p' \trans{a} p''$\ for some action $a$\ and test $t'$\ such that $t \trans{a} t'$. Therefore we have a computation prefix of the form
\[
p \barra t \shortrightarrow p_1 \barra t \shortrightarrow \cdots p_n \barra t \shortrightarrow \cdots \shortrightarrow p' \barra t \shortrightarrow p'' \barra t',
\]
hence the maximality of computation \eqref{eq:nonmaximalcomp}\ does not hold.\\

Further, condition (\ref{prf:cond1})\ ensures that a computation of $p \barra t$\ cannot have the form
\[
p \barra t \shortrightarrow p_1 \barra t \shortrightarrow \cdots \shortrightarrow  p_k \barra t \shortrightarrow p_{k+1} \barra t \shortrightarrow \cdots
\]

Therefore all computations of $p \barra t$\ have the form
\[
p \barra t \shortrightarrow p_1 \barra t \shortrightarrow \cdots \shortrightarrow p_n \barra t \shortrightarrow p' \barra t'
\]

with $p' \mustsatisfy t'$\ by condition (\ref{prf:cond3}); then for each computation of $p \barra t$ there exist $p'', t''$\ such that 
\[
p \barra t \shortrightarrow \cdots \shortrightarrow p' \barra t' \shortrightarrow \cdots \shortrightarrow p'' \barra t'',
\]

and $t''\trans{\omega}$. Hence, every computation from $p \barra t$\ is successful.
\end{proof}
\begin{prop}
\label{thm:minsubmust}
Assume the LTS of processes is branching finite. For every $ t \in T$,  $\rhomust(X_t) \subseteq \rhomin(X_t)$.
\end{prop}

\begin{proof}

We have to show $p \mustsatisfy t$\ implies $p \in \sem{\phi_t}$.\\
Suppose $p \mustsatisfy t$; since we are assuming that the set $T$, as well as the set $S$, contains only finite branching tests (processes), the maximal length of a successful computation fragment $|p, t|$\ is defined and finite.\\
Therefore we proceed by induction on $|p, t|$; the main technical property used is the Fixpoint Property \ref{thm:fixpointprop}(\ref{thm:fixprop}).
\begin{itemize}
\item $k = 0$: In this case, $t \trans{\omega}$, and hence for all $p \in S$\ we have $ p \mustsatisfy t$. 
Moreover, $\varphi_t = \ttt$, and hence for all $p \in S\; p \in \sem{\phi_t}\rhomin$,
\item $k > 0$. There are several cases to consider, according to the structure of the test $t$:
\begin{enumerate}
\item $t \nottrans{\omega}, t \nottrans{\tau}, t \longrightarrow$: we first show that $p \in \sem{\Acc {\{a | t \trans{a}}} \rhomin$.\\
Since $p \mustsatisfy t$, we have $p \Downarrow$. Consider a computation fragment of the form
\[
p \barra t \shortrightarrow \cdots \shortrightarrow p^n \barra t
\]

As $p \Downarrow$, we require that all computations rooted in $p^n \barra t$\ will eventually contain a term of the form $p^k \barra t'$, where $t' \neq t$. Further, as $t \nottrans{\tau}$, such a test should follow from a synchronisation between $p^{k-1}$\ and $t$. We have that then that, whenever $p\Trans{\tau} p^n$, there exists an action $a$\ such that $t \trans{a} t'$\ and $p^n \Trans{a} p^k$, which combined with the constraint $p \Downarrow$\ is equivalent to $p \in \sem{\Acc {\{a | t \trans{a}}}$.\\
We also have to show that $p \in \sem{[a]X_{t'}} \rhomin$. Let $p\trans{a}p'$. Then $p \mustsatisfy t$\ implies $p' \mustsatisfy t'$. Moreover, we have $|p', t'| < k$. By inductive hypothesis, we have that $p' \in \sem{\phi_{t'}}$, that is $p' \in \rhomin(X_{t'})$. Then the result $p \in \sem{[a]X_{t'}} \rhomin$\ holds.

\item $t \nottrans{\omega}, t \trans{\tau}$: A similar analysis as in the case above can be carried out.
\end{enumerate}
\end{itemize}
\end{proof}
\noindent
Combining these two propositions we get our second result. Let us say that a test $t$ 
from a LTS of tests $\mathcal{T} = \langle T, Act_{\tau}^{\omega}, \rightarrow\rangle$ 
is finitary if the derived LTS consisting of all states in $\mathcal{T}$ accessible from $t$ 
is finite.
\begin{thm}
\label{thm:musttest}
Assuming the LTS of processes is finite branching, every finitary test $t$ is\\ \emph{must}-representable. 
\end{thm}
\begin{proof}
  Consider any test $t$. We can apply Definition~\ref{def:tests} to the finite LTS of tests reachable from $t$ to obtain a formula $\phi_t$\ which \emph{must}-represents test $t$. Notice that this formula is not contained in $\rechml$, as it uses simultaneous least fixpoints. However, by Theorem \ref{thm:becik}\ there exists a formula $\phimust t \in \rechml$\ such that $\sem{\phi_t} = \sem{\phimust t}$, thus $t$\ is \emph{must}-representable. Further, since each operator used in Definition \ref{def:tests} to define $\varphi_t$\ belongs to \musthml, it is ensured that $\phimust t \in \musthml$.
\end{proof}
As  a Corollary we are able to show that \musthml is actually the largest language (up to
logical equivalence) of \emph{must}-testable formulae.
\begin{corollary}\label{cor:mustlargest}
Suppose $\phi$ is a formula in \rechml which is \emph{must}-testable. Then 
there exists some $\psi$ in $\musthml$ which is logically equivalent to it.
\end{corollary}

%\ifx\proofs\yes
\begin{proof}
Suppose $\phi$\ is \emph{must}-testable. By theorem \ref{thm:musthml} there exists a finite test $t = \Tmust \phi$\ which \emph{must}-represents $\phi$. Further, by theorem \ref{thm:musttest}\ there exists a formula $\psi = \phimust {t} \in \musthml$\ which \emph{must}-tests for $t$. Therefore
\[
p \in \sem{\phi} \Leftrightarrow p \mustsatisfy \Tmust \phi \Leftrightarrow p \in \sem{\psi}
\]
\end{proof}

\section{The may case}\label{sec:may}

In this paper we simply state the corresponding theorems for \emph{may} testing:
\begin{thm}
\label{thm:mayhml}
Suppose the LTS of processes is finite branching. Then for every 
$\phi \in \mayhml$, there exists a  test $\Tmay \phi$\ such that, $\phi$ \emph{may}-represents the test $\Tmay \phi$.
\end{thm}

\begin{thm}
\label{thm:maytest}
Assuming the LTS of processes is finite branching, every test $t$ is \emph{may}-representable.
\end{thm}

\begin{corollary}\label{cor:maylargest}
Suppose $\phi$ is a formula in \rechml which is \emph{may}-testable. Then 
there exist some $\psi$ in $\mayhml$ which is logically equivalent to it.
\end{corollary}

\begin{proof}
  Similar to that of Corollary~\ref{cor:mustlargest}.
\end{proof}

Our proofs for Theorem~\ref{thm:maytest} and Theorem~\ref{thm:mayhml}
are similar in style to the corresponding results for \emph{must}
testing, namely, namely Theorem~\ref{thm:musttest} and
Theorem~\ref{thm:musthml}.  Also , as we point out in the Conclusion,
they can be recovered by dualising the proofs of the corresponding
Theorems in \cite{aceto}.

\leaveout{
We now turn to the characterization of the \maysatisfy testing relation in terms of \rechml formulae.\\
First we will prove that each formula in \mayhml\\
\emph{may}-represents some test $t$ in grammar \ref{eq:tests}; then we show 
that if the LTS generated by a test $t$\ is finitely branching and finite state, 
then there exists a formula $\phi$\ which \emph{may}-represents $t$. In this case we do not require for the 
LTS of processes to be branching finite; this is not surprising, as informally speaking 
the definition of the \emph{may}-testing relation gives rise to a linear theory.

To prove that the power of tests defined in grammar \ref{eq:tests}\ can be captured (with respect to 
the \maysatisfy testing relation) by the language \mayhml, we will exploit the proof techniques used in \cite{aceto}.\\
To this end, we define the concept of weak satisfaction relation:

\begin{defi}
\label{def:wsatrel}
Let $R \subseteq S \times \mayhml$. Then $R$\ is a weak satisfaction relation if, and only if, it satisfies 
the following implications:
\begin{eqnarray*}
s \;R\; \ttt && \forall s \in S\\
s \;R\; \fff && \mbox{ for no } s \in S\\
s \;R\; \dmnd \alpha \phi &\Leftarrow& s\Trans{\alpha}s', s' \;R\; \phi \mbox{ for some } s' \in S\\
s \;R\; \phi_1 \vee \phi_2 &\Leftarrow& s\Trans{\tau}s', s' \;R\; \phi_1 \mbox{ or } s' \;R\; \phi_2
 \mbox{ for some } s' \in S\\
s \;R\; \lfp X \phi &\Leftarrow& s\Trans{\tau}s', s' \;R\; \phi\{\lfp X \phi /X \}\mbox{ for some } s' \in S
\end{eqnarray*}
\end{defi}

$\maymodels$ Is defined as the smallest weak satisfaction relation.\\
If the satisfaction relation $\models$\ is restricted to \mayhml, then it coincides with $\maymodels$. To prove 
this result, we will need to prove several properties the $\maymodels$\ and $\models$ relations.\\
First we prove the following result for $\maymodels$:
\begin{prop}
\label{prop:tauallowed}
Let $s \in S, \phi \in \mayhml$. Then $s \maymodels \phi$ if and only if 
there exists $s' \in S$ such that \\
$s\Trans{\tau}s', s' \maymodels \phi$.
\end{prop}
A similar result holds for the standard satisfaction relation $\models$.
\begin{prop}
\label{prop:mayhmltauclosed}
Let $p \in S$, $\phi \in \mayhml$. Then $p \models \phi$ if and only if there exists $p' : p \Trans{\tau} p'$\ and $p' \models \phi$.
\end{prop}

\begin{proof}
For the only if implication notice that for all $p \in S$\ $p\Trans{\tau} p$.\\
For the only if implication, notice that the semantics of \mayhml is defined on weak actions, and that $\sem{\dmnd \alpha \phi} = \sem{\dmnd \tau \dmnd \alpha \phi}$.
\end{proof}
Thus, for each closed formula $\phi in mayhml$, we have $\sem\phi = \sem{\dmnd \tau \phi}$.

\begin{prop}
\label{prop:mayeq}
$p \models \phi \mbox{ iff } p \maymodels \phi$.
\end{prop}
We are now ready to show that each formula of \mayhml \emph{may}-represents some test $t$.\\
For each formula $\phi$\ in \mayhml, the test $\Tmay \phi$\ is defined as below:
\begin{eqnarray*}
\Tmay \ttt &=& \omega.0\\
\Tmay \fff &=& 0\\
\Tmay X &=& X\\
\Tmay {\phi_1 \vee \phi_2} &=& \tau.\Tmay {\phi_1} + \tau.\Tmay{\phi_2}\\
\Tmay {\dmnd \alpha \phi} &=& \alpha.\Tmay \phi\\
\Tmay {\lfp X \phi} &=& \mu X.\Tmay{\phi}
\end{eqnarray*}

\begin{prop}
\label{prop:maywsr}
The relation $\Rmay = \{\;(s, \phi) \barra s \maysatisfy \Tmay \phi\}$\ is a weak satisfaction relation.
\end{prop}
\begin{proof}
We prove that $\Rmay$\ satisfies the implications of definition \ref{def:wsatrel}.
\begin{itemize}
\item $\Tmay \ttt = \omega.0$. It is trivial to check that each process in $S$\ \maysatisfy such a test.
\item $\Tmay \fff = 0$. Again, it is straightforward to show that for no process $p \in S$\ we have $p \maysatisfy \Tmay \fff$.
\item Suppose $p \Trans{a} p'$, and $p' \Rmay \phi$. Then, we have the computation prefix
\[
p \barra \alpha.\Tmay \phi \shortrightarrow \cdots \shortrightarrow p'' \barra \alpha.\Tmay \phi \shortrightarrow p' \barra \Tmay \phi\footnote{where $p'' = p'$\ in the case $\alpha = \tau$}.
\]

Since $p' \maysatisfy \Tmay \phi$ by the definition of $\Rmay$, the experiment $p \barra \Tmay {\dmnd a \phi}$\ has a successful computation, hence $p \Rmay \dmnd{\alpha}\phi$.
\item Suppose $p \Trans{\tau} p'$, and $p' \Rmay \phi_1$. Given an arbitrary formula $\phi_2$, consider the experiment\\
$p \barra \tau.\Tmay{\phi_1} + \tau.\Tmay{\phi_2}$, which has the computation fragment
\[
p \barra \tau.\Tmay{\phi_1} + \tau.\Tmay{\phi_2} \shortrightarrow p \barra \Tmay{\phi_1} \shortrightarrow \cdots \shortrightarrow p' \barra \Tmay{\phi_1}
\]

As $p' \maysatisfy \Tmay{\phi_1}$, we have $p \maysatisfy \Tmay{\phi_1 \vee \phi_2}$.

\item Suppose $p \Trans{\tau} p'$, with $p' \Rmay \psi \{\lfp X \psi/X\}$. We have two different cases to consider:
\begin{itemize}
\item If $X$\ does not appear free in $\psi$\ then $\Tmay{\lfp X \psi} = \Tmay \psi$. Moreover, we have $\psi\{\lfp X \psi/X\} = \psi$. Consider the computation fragment
\[
p \barra \Tmay \psi \shortrightarrow \cdots \shortrightarrow p' \barra \Tmay \psi
\]

Again, we have $p \Rmay \Tmay \psi$, which in this case is equivalent to $p \Rmay \Tmay{\lfp X \psi}$

\item If $X$ appears free in $\psi$, then $\Tmay{\lfp X \psi} = \mu X.\Tmay{\psi}$. In this case we have the computation
\[
p \barra \mu X.\Tmay{\psi}\shortrightarrow \cdots \shortrightarrow p' \barra \mu X.\Tmay{\psi} \shortrightarrow p' \barra \Tmay{\psi\{\mu X.\psi/X\}}
\]

and hence $p \Rmay \lfp X \psi$.
\end{itemize}
\end{itemize}
\end{proof}

\begin{prop}
\label{prop:maylwsr}
Let $p \in S$ and let $\phi \in \mayhml$. If $p \maysatisfy \Tmay \phi$\ then $p \maymodels \phi$.
\end{prop}

\begin{thm}
\label{thm:mayhml}
Let $\phi \in \mayhml$, $p \in S$. Then $p \models \phi$\ if and only if  $p \may \Tmay \phi$.
\end{thm}
\begin{proof}
For each $\phi \in \mayhml$. If $p \models \phi$, then $p \maymodels \phi$ (\ref{prop:mayeq}). Moreover, since $\maymodels$\ is the smallest weak satisfaction relation, we have $p \Rmay \phi$ (\ref{prop:maywsr}), therefore $p \may \Tmay \phi$.\\
For the opposite implication, suppose $p \may \Tmay{\phi}$. By Proposition \ref{prop:maylwsr}\ we have $p \maymodels \phi$, hence $p \models \phi$\ by Proposition \ref{prop:mayeq}.
\end{proof}

Next, we show that if the LTS of tests generated by a test is finite state, then the notion of \maysatisfy testing relation can be captured logically by formulae in \mayhml.\\
To show this result, we will use the same proof techniques involved in the proof of Theorem \ref{thm:musttest}.\\
First, assume to have a test indexed set of test variables $\{X_t\}$. Then, for each test $t$\ define the formula $\phi_t$\ as
\begin{eqnarray*}
\phi_t &= \ttt &\mbox {if }t\trans{\omega}\\
\phi_t &= \fff &\mbox {if }t \nottrans{\;}\\
%\phi_t &= \displaystyle{\bigvee_{t': t \trans{a} t'}} \dmnd {a} X_{t'} \vee \displaystyle{\bigvee_{t': t \trans{\tau} t'}} X_{t'}
\phi_t &= \displaystyle{\bigvee_{\alpha, t': t \trans{\alpha} t'}} \dmnd{\alpha} X_{t'} & \mbox{if } t \nottrans{\omega}, t\trans{\;}
\end{eqnarray*}

and take $\phimay t$\ to be $\rechml^+$\ formula $\min_t(X_T, \phi_T)$.\\
Next we define the following environments:
\begin{eqnarray*}
\rhomin &=& \sem{\phimay t}\\
\rhomay &=& \{ p \barra p \may t\}
\end{eqnarray*}

In the same style of section \ref{sec:must}, we will prove that the two environments above coincide.
\begin{prop}
\label{prop:minsubmay}
For each test $t, \rhomin(X_t) \subseteq \rhomay(X_t)$.
\end{prop}
\begin{proof}
Suppose the LTS generated by a test $t$\ is finite state. We just need to show $\sem{\phi}\rhomay \subseteq \rhomay(X_t)$. The result will then follow from an application of the Minimal Fixpoint Property \ref{thm:fixpointprop}(\ref{thm:minfixprop}). We perform a case analysis on the structure of the test $t$.
\begin{itemize}
\item $t \trans{\omega}$. In this case we have $\rhomay(X_t) = S$, so the statement trivially holds.
\item $t \nottrans{\;}$. W have $\phi_t = \fff$, hence $\sem{\phi_t} \rhomay = \emptyset$. Again, the statement is trivial.
\item $t \nottrans{\omega}, t \trans{\;}$. Suppose $p \in \sem{\phi}\rhomay$. We have that there exists at least one action $\alpha$\ such that $t \trans{\alpha} t'$ there exists a process $p'$\ such that $p \Trans{\alpha} p'$\ and $p' \maysatisfy t'$. Hence we have the computation fragment
\[
p \barra t \shortrightarrow \cdots \shortrightarrow p'' \barra t \shortrightarrow p' \barra t',
\]
so that $p \maysatisfy t$.
\end{itemize}
\end{proof} 

\begin{prop}
\label{thm:maytest}
For each test $t, \rhomay(X_t) \subseteq \rhomin(X_t)$.
\end{prop}

\begin{proof}
Again, assume the LTS generated by a test $t$\ is finite state.
Let $p$\ be a process such that $p \maysatisfy t$. We proceed by induction on the minimal length of a successful computation prefix $|p, t|$\ to show $p \in \rhomin(X_t)$. We will make use of the Fixpoint Property \ref{thm:fixpointprop}(\ref{thm:fixprop}), thus showing that $p \in \sem{\phi_t}\rhomin$.
\begin{itemize}
\item $|p, t| = 0$. In this case we have $t \trans{\omega}$. By definition, $\phi_t = \ttt$, so that we have $\sem{\phi_t}\rhomin = S$. This case is trivial.
\item $|p, t| > 0$. Let
\[
p \barra t \shortrightarrow p' \barra t' \shortrightarrow \cdots \shortrightarrow p_n \barra t_n
\]

be a successful computation prefix of length $|p, t|$. We distinguish several cases according to the structure of the computation. Since $p' \maysatisfy t'$ and $|p',t'|<|p,t|$, in each case we have $p' \in \sem{\phimay {t'}}\rhomin$\ by inductive hypothesis.
\begin{itemize}
\item $p = p'$, $t \trans{\tau} t'$; we have $p \in \sem{X_t}\rhomin = \sem{\dmnd{\tau}\phi_{t'}}\rhomin$. Then $p \in \sem{\bigvee_{\alpha, t':t\trans{\alpha}t'}\dmnd{\tau}X_{t'}}\rhomin$.
\item $p \trans{\tau} p'$, $t = t'$; we have $p' \in \sem{X_t}\rhomin$, and therefore $p \in \sem{X_t}\rhomin$ by Proposition \ref{prop:mayhmltauclosed}.
\item $p \trans{a} p'$, $t \trans{a} t'$; in this case $p \in \sem{\dmnd{a} X_{t'}}\rhomin$, and hence $p \in \sem{\phi_t}\rhomin$.
\end{itemize}
\end{itemize}
\end{proof}

\begin{corollary}\label{cor:maylargest}
Suppose $\phi$ is a formula in \rechml which is \emph{may}-testable. Then 
there exists some $\psi$ in $\mayhml$ which is logically equivalent to it.
\end{corollary}

\begin{proof}
Suppose $\phi$\ is \emph{may}-testable. By theorem \ref{thm:mayhml} there exists a finite test $t = \Tmay \phi$\ which \emph{may}-represents $\phi$. Further, by theorem \ref{thm:maytest}\ there exists a formula $\psi = \phimay {t} \in \mayhml$\ which \emph{may}-tests for $t$. Therefore
\[
p \in \sem{\phi} \Leftrightarrow p \maysatisfy \Tmay \phi \Leftrightarrow p \in \sem{\psi}
\]
\end{proof}
}

\section{Conclusions}\label{sec:end}

We have investigated the relationship between properties of processes as expressed 
in a recursive version of Hennessy-Milner logic, \rechml, and \emph{extensional} tests as 
defined in \cite{dhn}. In particular we have shown that both \emph{may} and \emph{must}
tests can be captured in the logic, and we have isolated logically complete sub-languages
of \rechml which can be captured by \emph{may} testing and \emph{must} testing. One 
consequence of these results is that the \emph{may} and \emph{must} testing preorders
of \cite{dhn} are determined by the logical properties in these sub-languages \mayhml and
\musthml respectively; however this is already a well-known result, \cite{Hennessy85}.

However these results come at the price of modifying the satisfaction relation; 
to satisfy a box formula a process is required to converge. One consequence of this 
change is that the language \rechml no longer characterises the standard notion
of \emph{weak bisimulation equivalence}, as this equivalence is insensitive to
divergence. But there are variations on \emph{bisimulation equivalence} which do
take divergence into account; see for example \cite{walker,cbl}.

The research reported here was initiated after reading \cite{aceto};
there a notion of testing was used which is different from both
\emph{may} and \emph{must} testing. They define $s$ \emph{passes} the
test $t$ whenever no computation from $s \;|\; t$ can perform the
success action $\omega$, and give a sub-language which characterises this form of testing. 
It is easy to check that $s$ \emph{passes} $t$ if and only if, in our terminology, 
$s$ \emph{may} $t$ is not true. So their notion of testing is dual to \emph{may} testing,
and therefore, not surprisingly, our results on \emph{may} testing are simply dual versions
of theirs. However we believe our results on \emph{must} testing, specifically Theorem~\ref{thm:musthml}
and Theorem~\ref{thm:musttest},  are new.  

We have concentrated on properties associated essentialy with the 
behavioural theory based on extensional testing. 
However there are a large number of other behavioural theories; 
see \cite{rob} for an extensive survey, including their characterisation in terms of
\emph{observational} properties.

\nocite{*}
\bibliographystyle{alpha}
\bibliography{tests}

\begin{thebibliography}{NYHO07}

\bibitem[Abr87]{abramsky}
S.~Abramsky.
\newblock Observation equivalence as a testing equivalence.
\newblock {\em Theoretical Computer Science}, 53:225--241, 1987.

\bibitem[AI99]{aceto}
Luca Aceto and Anna Ing{\'o}lfsd{\'o}ttir.
\newblock Testing hennessy-milner logic with recursion.
\newblock In Thomas \cite{DBLP:conf/fossacs/1999}, pages 41--55.

\bibitem[AILS07]{luca}
Luca Aceto, Anna Ing\'{o}lfsd\'{o}ttir, Kim~Guldstrand Larsen, and Jiri Srba.
\newblock {\em Reactive Systems: Modelling, Specification and Verification}.
\newblock Cambridge University Press, New York, NY, USA, 2007.

\bibitem[BJ89]{Boolos}
George~S. Boolos and Richard~C. Jeffrey.
\newblock {\em Computability and Logic}.
\newblock Cambridge University Press, third edition, 1989.

\bibitem[BRR87]{braureiroz87b}
W.~Brauer, W.~Reisig, and G.~Rozenberg, editors.
\newblock {\em Petri Nets: Applications and Relationships to Other Models of
  Concurrency}.
\newblock Number 255 in Lecture Notes in Computer Science. Springer-Verlag,
  1987.

\bibitem[CN78]{finiteapprox}
Bruno Courcelle and Maurice Nivat.
\newblock The algebraic semantics of recursive program schemes.
\newblock In Winkowski \cite{DBLP:conf/mfcs/1978}, pages 16--30.

\bibitem[DH84]{dhn}
R.~De{N}icola and M.~Hennessy.
\newblock Testing equivalences for processes.
\newblock {\em Theoretical Computer Science}, 24:83--113, 1984.

\bibitem[DN83]{rocco}
Rocco De~Nicola.
\newblock {A Complete Set of Axioms for a Theory of Communicating Sequential
  Processes}.
\newblock In {\em FCT}, pages 115--126, 1983.

\bibitem[Gla93]{rob}
Rob J.~van Glabbeek.
\newblock The linear time - branching time spectrum ii.
\newblock In {\em CONCUR '93: Proceedings of the 4th International Conference
  on Concurrency Theory}, pages 66--81, London, UK, 1993. Springer-Verlag.

\bibitem[Hen85]{Hennessy85}
M.~Hennessy.
\newblock Acceptance trees.
\newblock {\em Journal of the ACM}, 32(4):896--928, October 1985.

\bibitem[HM85]{hml}
Matthew Hennessy and Robin Milner.
\newblock Algebraic laws for nondeterminism and concurrency.
\newblock {\em J. ACM}, 32(1):137--161, 1985.

\bibitem[Hoa85]{csp}
C.A.R. Hoare.
\newblock {\em Communicating Sequential Processes}.
\newblock Prentice-Hall, 1985.

\bibitem[HP80]{cbl}
Matthew C.~B. Hennessy and Gordon~D. Plotkin.
\newblock A term model for {CCS}.
\newblock In {\em Mathematical Foundations of Computer Science 1980,
  Proceedings of the 9th Symposium}, volume~88 of {\em Lecture Notes in
  Computer Science}, pages 261--274, Rydzyna, Poland, 1--5~September 1980.
  Springer.

\bibitem[Mil89]{ccs}
R.~Milner.
\newblock {\em Communication and Concurrency}.
\newblock Prentice-Hall, 1989.

\bibitem[NV07]{vardi}
Sumit Nain and Moshe~Y. Vardi.
\newblock Branching vs. linear time: Semantical perspective.
\newblock In Namjoshi et~al. \cite{DBLP:conf/atva/2007}, pages 19--34.

\bibitem[NYHO07]{DBLP:conf/atva/2007}
Kedar~S. Namjoshi, Tomohiro Yoneda, Teruo Higashino, and Yoshio Okamura,
  editors.
\newblock {\em Automated Technology for Verification and Analysis, 5th
  International Symposium, ATVA 2007, Tokyo, Japan, October 22-25, 2007,
  Proceedings}, volume 4762 of {\em Lecture Notes in Computer Science}.
  Springer, 2007.

\bibitem[Old87]{olderog}
E.-R. Olderog.
\newblock Tcsp: Theory of communicating sequential processes.
\newblock In Brauer et~al. \cite{braureiroz87b}, pages 441--465.

\bibitem[RS96]{regulartrees}
G.~Rozenberg and A.~Salomaa, editors.
\newblock {\em Handbook of Formal Languages}, volume~3.
\newblock Springer Verlag, Berlin, Heidelberg, New York, October 1996.

\bibitem[Tho99]{DBLP:conf/fossacs/1999}
Wolfgang Thomas, editor.
\newblock {\em Foundations of Software Science and Computation Structure,
  Second International Conference, FoSSaCS'99, Held as Part of the European
  Joint Conferences on the Theory and Practice of Software, ETAPS'99,
  Amsterdam, The Netherlands, March 22-28, 1999, Proceedings}, volume 1578 of
  {\em Lecture Notes in Computer Science}. Springer, 1999.

\bibitem[Wal88]{walker}
David Walker.
\newblock Bisimulations and divergence.
\newblock In {\em Proceedings of the Third Annual IEEE Symposium on Logic in
  Computer Science (LICS 1988)}, pages 186--192. IEEE Computer Society Press,
  July 1988.

\bibitem[Win78]{DBLP:conf/mfcs/1978}
J{\'o}zef Winkowski, editor.
\newblock {\em Mathematical Foundations of Computer Science 1978, Proceedings,
  7th Symposium, Zakopane, Poland, September 4-8, 1978}, volume~64 of {\em
  Lecture Notes in Computer Science}. Springer, 1978.

\bibitem[Win93]{becik}
Glynn Winskel.
\newblock {\em The Formal Semantics of Programming Languages}.
\newblock The MIT Press, Cambrige, Massachusetts, 1993.

\end{thebibliography}
\leaveout{
\appendix
\section{Proofs of the propositions}
In this appendix we provide the proofs of all the propositions and theorems stated, but not proved, in the paper. Further, other well known results which are involved in the details of many proofs are stated, together with their proofs.

Below proofs of results contained in the paper are labelled with a reference to the proposition they refer to.
\begin{description}
 \item \begin{prop}\qquad
\label{prop:syntlemma}
\begin{enumerate}[(i)]
\item Let $\phi, \psi$\ be formulae such that $Y$\ does not occur free in $\psi$, and let $\rho$\ be an environment and $P \subseteq 2^S$. Then
\[
\sem{\phi}\rho[X \mapsto \sem{\psi} \rho][Y \mapsto P] = \sem{\phi}\rho[Y \mapsto P][X \mapsto \sem{\psi}\rho[Y \mapsto P]\,]
\]
\label{prop:substlemma}
\item Let $\phi, \psi \in \rechml$, and $\rho$\ be an environment: then 
\[
\sem{\phi\{\psi/X\}}\rho = \sem{\phi}\rho[X \mapsto \sem{\psi}\rho].
\]
\label{prop:envlemma}
\end{enumerate}
\end{prop}
\begin{proof}
Both proofs can be performed by induction on the structure of the formula $\phi$. For (\ref{prop:substlemma}) three different sub cases should be handled when dealing with the case $\phi \equiv Z$ (namely $Z \equiv X;\; Z\equiv Y$\ and $Z \not\equiv X, Z \not\equiv Y$).\\ This proposition is needed to prove (\ref{prop:envlemma}) when considering a fixpoint formula: specifically, let $\phi \equiv \lfp Y {\phi_1}$: without loss of generality we can assume $Y \not\equiv X$, so that $\lfp Y {\phi_1}\{\psi/X\} \equiv$\\
$\equiv \lfp Y {\phi_1\{\psi/X\}}$. By inductive hypothesis we have
\[
\sem{\phi_1\{\psi/X\}}\rho = \sem{\phi_1}\rho[X \mapsto \sem{\psi}\rho]
\]
and, therefore,
\begin{eqnarray*}
\sem{\lfp Y {\phi_1\{\psi/X\}}}\rho &=& \bigcap \{P : \sem{\phi_1\{\psi/X\}}\rho[Y \mapsto P] \subseteq P\}\\
&\iheq& \bigcap \{P : \sem{\phi_1}\rho[Y\mapsto P][X \mapsto \sem{\psi}\rho[Y\mapsto P]] \subseteq P\}\\
&\stackrel{(\scriptstyle{\ref{prop:substlemma}})}{=}& \bigcap \{P : \sem{\phi_1}\rho[X\mapsto \sem{\psi}\rho][Y \mapsto P] \subseteq P\}\\
&=& \sem{\lfp Y {\phi_1}}\rho[X \mapsto \sem{\psi}\rho]
\end{eqnarray*}
\end{proof}
\item[Theorem \ref{thm:becik}.]\begin{aproof}
\begin{enumerate}[(i)]
\item By straightforward computations: we will show only the case for $\slfp{1}{\overline{X}}{\overline{\phi}}$, as the other one is obtained by symmetry:
\begin{eqnarray*}
&\sem{\lfp {X_1} {\phi_1\{\lfp {X_2}{\phi_2}/X_2\}}}\rho &=\\
&\bigcap \{P:\ \sem{\phi_1\{\lfp {X_2}{\phi_2}/X_2\}}\rho[X \mapsto P] \subseteq P\}&\stackrel{\scriptstyle{\eqref{prop:syntlemma}}}{=}\\
&\bigcap\{P: \sem{\phi_1}\rho[X_1 \mapsto P][X_2 \mapsto \sem{\lfp {X_2}{\phi_2}}\rho[X_1 \mapsto P]] \subseteq P\}&=\\
&\bigcap \{P: \sem{\phi_1}\rho[X_1 \mapsto P][X_2 \mapsto \bigcap Q: \sem{\phi_2}\rho[X_1 \mapsto P][X_2 \mapsto Q] \subseteq Q\}] \subseteq P\}&=\\
&\pi_1 (\bigcap \{ \langle P, Q \rangle: \sem{\phi_2}\rho[X_1 \mapsto P][X_2 \mapsto Q] \subseteq Q, \sem{\phi_1}\rho[X_1 \mapsto P][X_2 \mapsto Q] \subseteq P\})&
\end{eqnarray*}
\item Let $n \geq 2$, and let $\slfp i {\overline{X}} {\overline{\phi}}$\ be a 
simultaneous fixpoint formula with $|\overline{X}| = |\overline{\phi}| = n$.\\
The proof of (\ref{prop:becik2}) can be used to 
obtain a logically equivalent formula of the form
$\slfp i {\overline{Y}} {\overline{\phi}}$, where\\ 
$|\overline{Y}| = |\overline{\psi}| = n-1$.\\
This procedure can thus be iterated until obtaining 
an equivalent formula $\slfp 1 {\langle Z\rangle} {\langle \varphi \rangle}$, 
which is equivalent to $\lfp Z \varphi \in \rechml$.
\end{enumerate}
\end{aproof}
\item[Theorem \ref{thm:fixpointprop}.]\begin{aproof}\qquad
\begin{enumerate}[(i)]
\item This follows from the definition of $\sem{\lfp {\overline{X}} {\overline{\phi}}}$. Let $\overline{P}$ be a vector of sets from $2^S$\ such that\\
$\sem{\phi_i} \rho[\overline{X} \mapsto \overline{P}] \subseteq P_i$. Then
\begin{eqnarray*}
\sem{\lfp {\overline{X}} {\overline{\phi}}}\rho &=& \bigcap \{ \overline{P} \;|\; \sem{\phi_i}\rho[\overline{X}\mapsto\overline{P}] \subseteq P_i,\; 1\leq i \leq n\}\\
&=& \overline{P} \cap \bigcap \{ \overline{P} \;|\; \sem{\phi_i}\rho[\overline{X}\mapsto\overline{P}] \subseteq P_i,\; 1\leq i \leq n\}
\end{eqnarray*}
we have therefore that
\[
\sem{\slfp i {\overline{X}} {\overline{P}}} = P_i \cap  \pi_i (\bigcap \{ \overline{P} \;|\; \sem{\phi_i}\rho[\overline{X}\mapsto\overline{P}] \subseteq P_i,\; 1\leq i \leq n\}) \subseteq P_i
\]

\item Let $1 \leq i \leq n$. By the definition of $\sem{\slfp i {\overline{X}} {\overline{\phi}}}$\ it holds 
\begin{eqnarray*}
\sem{\phi_i}\rhomin&=& \sem{\phi_i}\rho[\overline{X} \mapsto \sem{\lfp {\overline{X}} {\overline{\phi}}}\rho ]\\
&\subseteq& \sem{\slfp i {\overline{X}} {\overline{\phi}}}\rho\\
&=&\rhomin(X_i)
\end{eqnarray*}

The inclusion shows that $\sem{\phi_i}\rhomin \subseteq \sem{\slfp i {\overline{X}} {\overline{\phi}}}\rho$. Moreover, since $\sem{\phi_i}\rhomin \subseteq \rhomin$, the converse inclusion follows from (\ref{thm:minfixprop})

\end{enumerate}
\end{aproof}

\item \begin{corollary}
\label{cor:minsubst}
Let $\phi \equiv \lfp X \psi$ be a formula in \rechml. Then $\phi$\ is logically equivalent to $\psi\{\lfp X \psi/X\}$.
\end{corollary}
\begin{proof}
Given a closed formula $\phi \equiv \lfp X \psi$\ and an arbitrary environment $\rho$, we have\\
$\sem{\lfp X \psi} = \sem {\psi}\rho[X \mapsto \sem{\lfp X \psi}]$\ 
by an application of Theorem \ref{thm:fixpointprop}(\ref{thm:fixprop}). 
Further,\\ $\sem{\psi}\rho[X \mapsto \sem{\lfp X \psi}] = \sem{\psi\{\lfp X \phi\}}$\ 
by Proposition \ref{prop:syntlemma}(\ref{prop:envlemma}).
\end{proof}

\item \begin{thm}[\cite{becik}]
\label{thm:tarski}
Let $\phi \equiv \lfp X \psi$\ a formula in \rechml. Then $\sem{\phi}$\ is the least solution of the equation
\[
X = \psi
\]
\end{thm}
\begin{proof}

Corollary \ref{cor:minsubst}\ ensures that $\sem{\phi}$\ is a solution of the equation $X = \psi$. Moreover, let $P$\ be a solution to such an equation; we have
\[
\sem{\psi}[X \mapsto P] = P,
\]

therefore $P \in \{ P \;|\; \sem{\psi}[X \mapsto P] \subseteq P\}$. Now it is trivial to notice $\sem{\lfp X \psi} \subseteq P$.
\end{proof}

\leaveout{\item \begin{lem}
\label{lem:konig}
Suppose both the LTS of processes and the LTS of tests are branching finite; for each process $p$\ and test $t$\ we have:
\begin{enumerate}[(a)]
\item The LTS of experiments $\mathcal{E} = \langle P \barra T, Act_{\tau} \cup {\omega}, \shortrightarrow \rangle$\ is branching finite.
\label{lem:konig1}
\item If $p \mustsatisfy t$\ then each minimal successful computation prefix of the experiment $e = p\;|\;t$ has the form
\[
e = e_0 \shortrightarrow e_1 \shortrightarrow \cdots e_n
\]
where $e_n \trans{\omega}$\ and $e_i = e_j$\ implies $i = j$ for all $0 \leq i, j \leq n$.
\label{lem:konig2}
\item If $p \mustsatisfy t$\ the maximal length of a successful computation of $|p, t|$\ is defined as
\[
|p, t| = \begin{cases}
0 & \mbox{ if } t \trans{\omega}\\
\max \{ \;|p',t'| \;:\; (p \barra t) \shortrightarrow (p' \barra t')\;\} + 1 & \mbox{ otherwise}
\end{cases}
\]
\label{lem:konig3}
\item If $p \mustsatisfy t$\ then |p, t|\ is finite.
\label{lem:konig4}
\end{enumerate}
\end{lem}
\begin{proof}\qquad
\begin{enumerate}[(a)]
\item Suppose $e \shortrightarrow e'$; then $e' = p' \barra t'$, where one of the following occurs:
\begin{itemize}
\item $p \trans{\tau} p'$, $t = t'$,
\item $p = p'$, $t \trans{\tau} t'$,
\item $p \trans{a} p'$, $t \trans{a} t'$.
\end{itemize}

The number of successors of $p \barra t$\ is given by
\begin{eqnarray*}
|\Succ{p \barra t}| &=& |\Succ{\tau, p}| + |\Succ{\tau, t}| + \sum_{a \in Act} |\Succ{a, p}| \cdot |\Succ{a, t}|\\
&\leq& |\Succ{\tau, p}| + |\Succ{\tau, t}| + (\sum_{a \in Act} |\Succ{a, p}|\;) \cdot (\sum_{a \in Act} |\Succ{a, t}|\;)\\
&\leq& |\Succ{\tau, p}| + |\Succ{\tau, t}| + |\Succ p|\cdot|\Succ t|
\end{eqnarray*}

As both $p, t$\ are branching finite, each amount in the last sum is finite, and so is $|\Succ{p \barra t}|$.

\item We prove the contrapositive statement: suppose that an experiment $e$\ has a minimal successful computation prefix such that
\[
e = e_0 \shortrightarrow \cdots \shortrightarrow e_n
\]

and there exist two indexes $0 \leq i, j \leq n $\ such that $e_i = e_j$ and $i \neq j$. Without loss of generality, let $i <j$; as we are assuming  that the successful computation above is minimal, it cannot be $j = n$. We have the (infinite) computation
\[
e_0 \shortrightarrow \cdots e_i \shortrightarrow e_{i+1} \shortrightarrow \cdots \shortrightarrow e_j = e_i \shortrightarrow e_{i+1}\shortrightarrow \cdots
\]

which is unsuccessful. Hence $p$\ does not \mustsatisfy\ $t$.
\item If $t \trans{\omega}$\ then the maximal computation prefix of a computation of $p \barra t$\ is given by $p \barra t$, hence $|p, t| = 0$.\\
If $t \nottrans{\omega}$\ and $p \mustsatisfy t$, then each computation prefix of the form
\[
p \barra t \shortrightarrow p' \barra t'
\]

is successful, with maximal length prefix computation $|p', t'|$. This value is less or equal to $\max \{|p', t'|\;:\; p \barra t \shortrightarrow p' \barra t'\}$; moreover, there exists at least one successor of $p \barra t$\ for which equality holds. Hence,
\[
|p, t| = \max \{ |p', t'| \; :\; p \barra t \shortrightarrow p' \barra t'\} + 1
\]

\item Suppose $p \mustsatisfy t$. By (\ref{lem:konig1}) $p \barra t$\ is finite branching, id est the set $\{p' \barra t' \;;\; p \barra t \shortrightarrow p' \barra t'\}$\ is finite.\\
Further, by (\ref{lem:konig2}) we have an inductive definition of $|p, t|$, which is therefore well defined and finite.
\end{enumerate}
\end{proof}
}
\item[Proposition \ref{prop:Tbf}]\begin{aproof}\qquad
\begin{enumerate}[(i)]
\item Each term of grammar \eqref{eq:tests}\ can be represented as
\[
\sum_{i \in I} t_i
\]

where $I$\ is finite and each $t_i$\ is either in the form $0$, $\alpha.t'$\ or $\mu X.t'$. Then for each $i \in I$\ the number of outgoing transitions $n(t_i)$\ of $t_i$\ is at most one: we have therefore
\[
n(t) \leq \sum_{i \in I} n(t_i) \leq |I|
\]

This argument applies to all states of the generated LTS: hence $\mathcal{T}$\ is branching finite.

\item A standard proof of this Proposition can be obtained by converting each test in a \textbf{Nondeterministic Finite state Tree Automata} \cite{regulartrees}.
However, since such topics are beyond the scope of this report, we will provide a simpler proof.

The main role of the syntactic replacement that occurs after a transition of the form $\mu X.t \trans{\tau} t\{\mu X.t/X\}$\ is to add loops in the LTS generated by $\mu X.t$\ from a state to the state $\mu X.t$\ itself. For each test $t$, a loop free test $\underline{t}$\ can be defined as
\begin{eqnarray*}
\underline{0} &=& 0\\
\underline{\alpha.t} &=& \alpha.\underline{t}\\
\underline{\omega.t} &=& \omega.\underline{t}\\
\underline{t_1 + t_2} &=& \underline{t_1} + \underline{t_2}\\
\underline{\mu X.t} &=& \tau.\underline{t\{0/X\}}
\end{eqnarray*}

Notice that, for every test $t$ generated by Grammar \eqref{eq:tests}, the test $t\{0/X\}$ is generated by 
the same Grammar, so that $\underline{t}$ is well defined.
It is straightforward to prove that the number of states of the LTS generated by $t$\ are at most the number of states of the LTS generated by $\underline{t}$\footnote{equality holds if and only if $t$\ does not contain any sub term of the form $\mu X.t'$, where $X$\ appears free in $t'$.}\\
Since a test $\underline{t}$\ does not contain any formula of the form $\mu X.t'$, it is possible to prove that the the LTS generated by $\underline{t}$\ is finite state by a straightforward induction on the structure of $\underline{t}$\ itself. Therefore, it follows that also the LTS generated by an arbitrary test $t$\ is finite state.
\end{enumerate}
\end{aproof}
\item[Lemma \ref{lem:divergence}. ]\begin{aproof}
Let $p$ be a process such that $p \Uparrow$, let $\phi \in \musthml$ such that $p \in \sem \phi$.
Then $\phi$ cannot be $\Acc A, \fff, [\alpha] \phi$\ 
nor a conjunction of formulae containing one of such terms.\\
We now show that $\phi$ cannot be a formula of the form $\lfp X \psi$, where $\psi$ 
contains either free occurrences of the variable $X$ or the operators $\Acc A, \bbox \alpha$. 
To this end, we perform a case analysis on the formul $\psi$: 
\begin{enumerate}[(i)]
\item \label{lem:div1} $\psi$ contains an occurrence of the operator $\bbox \alpha$. Here we can apply Corollary 
\ref{cor:minsubst} to obtain a formula of the form 
$\bbox \alpha \phi' \wedge \phi''$ which is logically equivalent to $\phi$. Thus, if $p \Uparrow$ then 
$p \notin \sem \phi$, 
\item \label{lem:div2} $\psi$ contains the operator $\Acc A$. We can proceed as in Case \ref{lem:div1}, 
\item \label{lem:div3} $\psi$\ contains at least a free occurrence of 
variable $X$. If such an occurrence is guarded by a $\bbox \alpha$ operator, then we can proceed as in 
Case \ref{lem:div1}. Otherwise we can obtain a formula of the form $\lfp X {X \wedge \psi'}$ which is 
equivalent to $\phi = \lfp X \psi$. Again, this is done by a repeated application of Corollary \ref{cor:minsubst}. 
Now it is trivial to notice that $\emptyset$ is a solution to the equation $X = X \wedge \psi$, and therefore 
it is its least solution. Hence $\sem{\phi} = \emptyset$, so that $p \notin \sem{\phi}$.
\end{enumerate}

The only possible case left for $p\Uparrow$, $p \in \sem\phi$ to hold is therefore 
given by $\phi$ being generated by the Grammar below:
\begin{equation}
\phi \is \ttt \barra \phi_1 \wedge \phi_2 \barra \lfp X \phi.
\label{eq:ttgrammar}
\end{equation}
It is trivial now to show $\sem{\phi} = S$.
\end{aproof}
\leaveout{
\item[Proposition \ref{prop:cpo}.]\begin{aproof}
Let $X_0 \subseteq X_1 \subseteq \cdots$ be a chain of 
elements from $\mathcal{C}$. If $X_i = S$ for some $i \geq 0$, 
then $\bigcup_{n} X_n = S$, which is therefore included in 
$\mathcal{C}$.\\
Suppose then $X_i \neq S$ for all $i \geq 0$. Thus we have 
that, whenever $p \in \bigcup_n X_n$ then $p \in X_i$ for 
some $i$, and therefore $p \Downarrow$. Then $\bigcup_n 
X_n \in \mathcal{C}$.
\end{aproof}
}
\item[Proposition \ref{prop:dmndcontinuous}.]\begin{aproof}
 It is trivial to show that
\[
\bigcup_n[\cdot \alpha \cdot] X_n \subseteq [\cdot \alpha \cdot] \bigcup_n X_n.
\]

Thus we only need to show that the opposite implication holds.\\
First, notice that it $X_i = S$ for some $i$, then
\[
\bigcup_n\bbox{\cdot\alpha\cdot}X_n = \{ s :\; s \Downarrow\} = \bbox{\cdot\alpha\cdot}\bigcup_n X_n
\]

Suppose then that $X_i \neq S$ for all $i \geq 0$. Then we have 
$\bigcup_n X_n \neq S$.
By definition the set $\bbox{\cdot\alpha\cdot}\bigcup_n X_n$ can 
be written as
\[
\{s\;:\; s \Downarrow, \Succ{\alpha, s} \subseteq \bigcup_n X_n\}.
\]
 We will prove that for each state $s$ in such a set $\Succ{\alpha, s}$ 
is finite, therefore there exists an $X_n$ such that $\Succ{\alpha, s} 
\subseteq X_n$. As a direct consequence, $s \in \bbox{\cdot\alpha\cdot} X_n$, 
which is included in $\bigcup_n \bbox{\cdot\alpha\cdot}X_n$.

Let $s \in \bbox{\cdot\alpha\cdot} \bigcup_n X_n$ and let 
$s'$ be one of its $\alpha$ derivative. By definition we have 
$s' \in \bigcup_n X_n$. Thus there exists $n \geq 0$ such that 
$s' \in X_n$. Since $X_n \in \mathcal{C}$, $X_n \neq S$, it holds 
$s' \Downarrow$. Since we are assuming that the LTS of processes 
is finite, as a consequence of Konig's lemma we obtain 
that if the set $\Succ{\alpha, s}$ is infinite then  
the $\tau$-computation tree of either $s$ or one of 
its $\alpha$-derivative $s'$ has an infinite path. 
The former contradicts the statement $s\Downarrow$, while 
the latter contradicts the property $s' \Downarrow$ we just proved. 
Thus $\Succ{\alpha, s}$ is finite.
\end{aproof}

\item[Lemma \ref{lem:statespaceformulae}.]\begin{aproof} Suppose $\sem{\phi} = S$ and let $p$\ be a process such that $p \Uparrow$; since $p \in \sem{\phi}$, the same argument used for the proof of Lemma \ref{lem:divergence} applies.
\end{aproof}

\item[Lemma \ref{lem:testprops}.]\begin{aproof}\qquad
\begin{itemize}
\item Suppose $p \mustsatisfy \mu X.t$. Then all computations with prefix
\[
p \;|\; \mu X.t \shortrightarrow p \;|\; t\{\mu X.t/X\}
\]

are successful: hence $p \mustsatisfy t\{\mu X.t/X\}$.
\item Suppose $p \Downarrow, p \mustsatisfy t\{\mu X.t/X\}$. Then for each computation of $p \barra \mu X.t$\ with prefix
\[
p \;|\; \mu X.t \shortrightarrow \cdots \shortrightarrow p' \;|\; \mu X.t \shortrightarrow p' \;|\; t\{\mu X.t/X\}
\]
there exists a computation with prefix
\[
p \;|\; t\{\mu X.t/X\} \shortrightarrow \cdots \shortrightarrow p' \;|\; t\{\mu X.t/X\}
\]
which is successful. Hence $p \mustsatisfy \mu X.t$.
\end{itemize}
\end{aproof}
\item \begin{lem}
\label{lem:formulae2tests}
Let $\phi, \psi$\ be formulae in \musthml. Then
\[
\Tmust{\phi\{\psi/X\}} \equiv \Tmust \phi \{\Tmust \psi/X\}
\]

where $\equiv$\ is the syntactic congruence equivalence over tests.
\end{lem}
\begin{proof}
By a straightforward induction on the structure of $\phi$, providing a consistent renaming of bounded variables when dealing with the case $\phi = \lfp X \varphi$, with $\varphi$\ not closed.
\end{proof}
\item[Proposition \ref{prop:satisfaction}.]\begin{aproof}
The definition of $\sem\cdot$\ ensures that $\models$\ is a satisfaction relation; we have:
\begin{eqnarray*}
(\models \ttt) &=& S\\
(\models \fff) &=& \emptyset\\
(\models \Acc A) &=& \{ \setof{s}{s \Downarrow, s \Trans{\tau} s' \mbox{ implies } S(s') \cap A \neq \emptyset}\\
(\models \;[\alpha]\phi) &=& [\cdot \alpha \cdot] (\models\; \phi)\\
(\models \;\phi\{\lfp X \phi /X\}) &=& (\models\; \lfp X \phi)
\end{eqnarray*}
where the last equality follows from Corollary \ref{cor:minsubst}.

It remains to show that $\models$\ is in fact the smallest satisfaction relation.\\
Let $R$\ be a satisfaction relation, and suppose that $p \in \sem{\phi}$: we show that $p \; R \; \phi$.\\
By Proposition \ref{cor:continuity}\ there exists $k \geq 0$\ such that $p \in \sem{\phi^k}$. 
We proceed by induction on $k$.\\
The case $k = 0$\ is vacuous. Assume the result holds for a generic $k$; 
we will perform an inner induction on the structure of $\phi$. 
Again, only the most interesting details are given.\\
Suppose $\phi = \lfp X \psi$: then $\lfp{X}{\psi}^{(k+1)} = (\psi\{\phi/X\})^k$, 
and by inductive hypothesis $p \;R\; \psi\{\phi/X\}$ follows, and so $p \; R \; \phi$ 
by the definition of satisfaction relation.\\
Finally, if $\phi$ has the form $\bbox{\alpha}\psi$ or $\phi_1 \wedge \phi_2$, it is 
not possible to use the inductive hypothesis directly. This is because
$(\bbox{\alpha}\phi)^{(k+1)} = \bbox \alpha (\phi)^{(k+1)}, 
(\phi_1 \wedge \phi_2)^{(k+1)} = \phi_1{(k+1)} \wedge \phi_2^{(k+1)}$.\\
We define therefore the height of a formula $h(\phi)$ as
\begin{eqnarray*}
h(\ttt) &=& 0\\
h(\fff) &=& 0\\
h(\Acc A) &=& 0\\
h(\lfp X \psi) &=& 0\\
h(\bbox \alpha \psi) &=& h(\psi) + 1\\
h(\phi_1 \wedge \phi_2) &=& \mbox{max}(h(\phi_1), h(\phi_2)) +1
\end{eqnarray*}
and we perform another induction of $h(\phi)$. The case $h(\phi) = 0$ 
has already been handled. Suppose then $h(\phi) = n+1$; then 
either $\phi = \bbox \alpha \psi$ or $\phi = \phi_1 \wedge \phi_2$. 
We will consider only the first case.Here $h(\psi) = n$, so that 
by inductive hypothesis we have $p' \models \psi$ implies $p'\;R\;\psi$.\\
If $p \models \bbox \alpha \psi$ then $p \Downarrow$; further, whenever 
$p \Trans{\alpha} p'$, we have $p' \models \psi$ and therefore $p' \;R\; \psi$. 
Thus $p \in \bbox{\cdot \alpha \cdot}(R \phi)$.
\leaveout{
Notice that, when dealing with a formula of 
the form $[\alpha]\phi$, another induction on the number of the leftmost $[\alpha]$\ operators is needed, 
as $([\alpha]\phi)^{(k+1)} = [\alpha](\phi^{(k+1)})$. 
A similar analysis has to be made when dealing with conjunctions.
}
\end{aproof}
\leaveout{\item[Proposition \ref{prop:tauallowed}.]\begin{aproof}
\begin{description}
 \item \item[If: ] Since $\maymodels$\ is the smallest weak satisfaction relation, it satisfies both the \\
 implications of definition \ref{def:wsatrel}\ and their converse. Suppose $s \trans{\tau} s'$, and 
$s' \maymodels \phi$; we proceed by induction on $\phi$:
\begin{itemize}
 \item $\phi \equiv \ttt$. By definition $s \maymodels \ttt$.
 \item $\phi \equiv \fff$. Vacuous, as $s' \not\maymodels \fff$.
 \item $\phi \equiv \dmnd \alpha \psi$. As $s' \maymodels \dmnd \alpha \psi$, there exists $s''$\ such that 
  $s' \Trans{\alpha} s''$\ and $s'' \maymodels \phi$. Since $s \Trans{\tau} s'$\ by hypothesis, 
  $s \Trans{\alpha} s''$, and by definition \ref{def:wsatrel}\ $s \maymodels \dmnd \alpha \psi$.
 \item $\phi \equiv \phi_1 \vee \phi_2$. By hypothesis, $s' \Trans{\tau} s''$\ such that $s'' \maymodels \phi_1$ or 
  $s'' \maymodels \phi_2$; without loss of generality, assume $s'' \maymodels \phi_1$. As $s \Trans{\tau} s'$, 
  $s \Trans{\tau} s''$, and by definition $s \maymodels \phi_1 \vee \phi_2$.
 \item $\phi \equiv \lfp X \psi$. As $s' \maymodels \lfp X \psi$, there exists $s''$\ such that 
  $s' \Trans{\tau} s'', s'' \maymodels \psi \{ \lfp X \psi / X \}$. Again, we have $s \Trans{\tau} s''$, 
  therefore $s\maymodels \lfp X \psi$.
\end{itemize}

\item[Only if: ] We just need to show that the relation
\[
 R = \{\;(s, \phi)\}\;|\;s\Trans{\tau}s',\;s'\maymodels\phi \mbox{ for some } s' \in S\;\}
\]

is a weak satisfaction relation; then we remember that $\maymodels$\ is the smallest weak satisfaction relation, and 
therefore is included in $R$, to claim the validity of the implication. We proceed by cases on $\phi$:
\begin{itemize}
 \item $\phi \equiv \ttt$. We have $s \maymodels \ttt$, $s \Trans{\tau} s$ for all $s \in S$. Then $s\;R\;\ttt$\ 
 is true for all $s \in S$.
 \item $\phi \equiv \fff$. We have $s' \maymodels \fff$\ for no $s' \in S$. Therefore, there is no $s$\ such that 
 $s \Trans{\tau} s'$\ and $s' \maymodels \fff$; thus for no $s \in S\; s \maymodels \fff$.
 \item $\phi \equiv \dmnd \alpha \psi$. Suppose $s \Trans{\alpha} s', s'\;R\;\psi$; we have to show 
 $s\;R\;\dmnd \alpha \psi$.\\
 As $s'\;R\;\psi$, by definition there exists $s''$\ such that $s' \Trans{\tau}s''$\ and $s''\maymodels \psi$. 
 Therefore $s \Trans{\alpha} s''$, and since $s \Trans{\tau} s$\ we have $s \; R\; \dmnd \alpha \psi$.
 \item $\phi \equiv \phi_1 \vee \phi_2$. Without loss of generality, suppose 
 $s \trans{\tau} s', s' ;R\; \phi_1$; we have to show $s\;R\; \phi_1 \vee \phi_2$. By definition of $R$,
 there exists $s''$ such that $s' \Trans{\tau} s''$, $s'' \maymodels \phi_1$. 
 Therefore $s' \maymodels \phi_1 \vee \phi_2$, and since $s \Trans{\tau} s'$\ we have $s \;R\; \phi_1 \vee \phi_2$.
 \item $\phi \equiv \lfp X \psi$. Suppose $s \Trans{\tau} s'$, $s' \;R\; \psi\{\lfp X \psi / X \}$; we have to show 
 $s \; R \; \lfp X \psi$. By definition of $R$\ there exists $s''$\ such that $s' \Trans{\tau} s''$\ and 
 $s'' \maymodels \psi \{ \lfp X \psi / X\}$. By definition of $\maymodels$\ this is the same of
 $ s \Trans{\tau} s'$, and $s' \maymodels \lfp X \psi$, or equivalently $s\; R\; \lfp X \psi$.
 
\end{itemize}
\end{description}
\end{aproof}
\item[Proposition \ref{prop:mayeq}.]\begin{aproof}
\begin{description}
\item[If: ] We just need to show that $\models$\ is a weak satisfaction relation. This is straightforward, as by Proposition \ref{prop:mayhmltauclosed}\ and the definition of $\sem\cdot$\ we have the following implications:
\begin{eqnarray*}
p \models \ttt && \mbox{ for all } p \in S\\
p \models \fff && \mbox{ for no } p \in S\\
p \models \dmnd \alpha \phi &\mbox{iff}& p \Trans{\alpha} p' \mbox{ for some } p' \in P\\
p \models \phi_1 \vee \phi_2 &\mbox{iff}& p \Trans{\tau} p' \mbox{ for some } p' \mbox{ such that } p' \models \phi_1 \vee \phi_2\\
p \models \lfp X \phi &\mbox{iff}& p \Trans{\tau} p' \mbox{ for some } p' \mbox{ such that} p' \models \phi\{\lfp X \phi/X\}
\end{eqnarray*}
where the last implication is obtained by Corollary \ref{cor:minsubst}.
\item{Only if: } Suppose $p \models \phi$. We show $p \models \phi$\ by induction on $\phi$. The only interesting case is given by $\phi \equiv \lfp X \psi$.\\
As $p \models \lfp X \psi$, by Proposition \ref{prop:mayhmltauclosed} and Corollary \ref{cor:minsubst} there exists $p'$\ with $p\Trans{\tau}p'$\ and $p'\models \psi\{\lfp X \psi/X\}$. We need to show $p' \maymodels \psi\{\lfp X \psi/X\}$, then we obtain $p \maymodels \lfp X \phi$\ by the definition of $\maymodels$.
Remember that $p \maymodels \lfp X \psi$\ if and only if $p \Trans{\tau} p'$\ for some $p'$\ such that $p' \maymodels \psi\{\lfp X \psi\}$; further, by Proposition \ref{prop:tauallowed}\ this is true if and only if $\phi \maymodels \phi\{\lfp X \psi\}$. Therefore, the set
\[
P = \{ p \;|\; p \maymodels \lfp X \psi\}
\]
is a solution to the equation $X = \psi$. Finally, we apply Tarski's fixed point Theorem \ref{thm:tarski}\ to claim that $\sem{\lfp X \psi}$\ is a subset of $P$, therefore $p \models \lfp X \psi$\ implies $p \maymodels \lfp X \psi$.
\end{description}
\end{aproof}
\item[Proposition \ref{prop:maylwsr}.] \begin{aproof}
Assume $p \must \Tmay \phi$. We proceed by induction on the minimal length of a successful prefix of a computation, denoted $|p, \Tmay \phi|$\ with an abuse of notation, to show that $p \maymodels \phi$.
\begin{itemize}
\item $|p, \Tmay \phi| = 0$. Then we may infer $\Tmay \phi \trans{\omega}$\ hence $\phi \equiv \ttt$. 
In this case, for each $p\in S$\ it holds. $p \maysatisfy \Tmay \phi$, and $\forall p \in S. p \maymodels \ttt$.
\item $|p, \Tmay \phi| = k+1$. Assume the statement holds for $k$, and consider the prefix 
\[
p | \Tmay \phi \shortrightarrow p' | t'
\]

of a minimal successful computation.\\ We distinguish several cases:
\begin{enumerate}[(a)]
\item $p \trans{\tau} p', t' \equiv \Tmay \phi$. Then by inductive hypothesis $p' \maymodels \phi$, 
and by proposition \ref{prop:tauallowed}\ we have $p \maymodels \phi$.
\item $p = p', \Tmay \phi \trans{\tau} t'$: in this case there are two possibilities.
\begin{itemize}
\item $\phi = \lfp X \psi$\ for some $\psi$. Hence $t' \equiv \Tmay \psi \{\Tmay \phi/X\}$, 
which is $t' \equiv \Tmay {\psi \{\phi/X\}}$. Again, by induction we have 
$p \maymodels \psi\{\phi/X\}$, and hence $p \maymodels \phi$.
\item $\phi = \phi_1 \vee \phi_2$. Without loss of generality 
we may infer $t' \equiv \Tmay {\phi_1}$. By Inductive hypothesis 
we have $p \maymodels \phi_1$, hence $p \maymodels \phi_1 \vee \phi_2$.
\end{itemize}
\item $p \trans{a} p', \Tmay \phi \trans{a} t'$. In this case we have 
$\phi = \dmnd \alpha \psi$, and hence $t' \equiv \Tmay \psi$. 
Then, by using the inductive hypothesis again, we have $p \maymodels \dmnd a \phi$.
\end{enumerate}
\end{itemize}
\end{aproof}
}
\end{description}
}
\end{document}